%% file: main.tex
\documentclass[runningheads]{llncs}

\usepackage[T1]{fontenc}
\usepackage{lmodern}
\usepackage{amsmath,amssymb,bm}
\usepackage{booktabs}
\usepackage{multirow}
\usepackage{array}
\usepackage{enumitem}
\usepackage[linesnumbered,ruled,vlined]{algorithm2e}
\usepackage{graphicx}
\usepackage{float}
\usepackage{wrapfig}
\usepackage{tikz}
\usetikzlibrary{arrows.meta,automata,positioning,calc}
\usepackage{pgfplots}
\pgfplotsset{compat=1.18}
\usepackage{xcolor}
\usepackage[colorlinks=true,linkcolor=blue,citecolor=blue,urlcolor=blue]{hyperref}
\usepackage{orcidlink}
\usepackage{marvosym}

\input{commands}

\begin{document}

\title{Compositional Automata Learning of \\Timed Synchronous Systems}
\titlerunning{Compositional Automata Learning of Timed Synchronous Systems}

\author{
    Hanyue Chen\,\orcidlink{0009-0008-9587-0178} \and
    Chang Liu\,\orcidlink{0009-0008-0518-0676} \and
    Miaomiao Zhang\textsuperscript{(\Letter)}  
} 
\authorrunning{H. Chen et al.}
\institute{School of Computer Science and Technology, Tongji University, Shanghai, China\\
\email{\{2111285,2433258,miaomiao\}@tongji.edu.cn}
}

\maketitle

\begin{abstract}
\input{Sections/0-Abstract}


\end{abstract}

\input{Sections/1-Introduction}

\input{Sections/2-Preliminaries}
\input{Sections/3-Problem_Formulation}
\input{Sections/4-Query_Adapter}
\input{Sections/5-Component_Learner}

\input{Sections/6-Algorithm}

\input{Sections/7-Experiment}
\input{Sections/8-Related_Work}

\input{Sections/9-Conclusion}

\bibliographystyle{splncs04}
\bibliography{reference}

\clearpage
\appendix
\input{Sections/Appendix}

\end{document}

%% file: commands.tex
\newcommand{\AlgorithmFont}{\fontsize{8pt}{9pt}\selectfont}
\SetAlFnt{\AlgorithmFont}

\SetAlgoSkip{AlgoOuterSkip}
\makeatletter
\newenvironment{compactalgorithm}[2][htbp]
{%
  \let\compactalgorithm@finish\@algocf@finish
  \def\@algocf@finish{%
    \compactalgorithm@finish
    \vspace{#2}%
  }%
  \begin{algorithm}[#1]%
}
{%
  \end{algorithm}%
}
\newenvironment{compacttable}[2][htbp]
{%
  \def\compacttable@bottomskip{#2}%
  \begin{table}[#1]%
}
{%
  \par\vspace{\compacttable@bottomskip}%
  \end{table}%
}
\makeatother

\newcommand{\project}[2]{{#1}{\downharpoonleft_{#2}}}
\newcommand{\rabs}[1]{\mathopen{[\![}{#1}\mathclose{]\!]}}

\newcommand{\NR}{\mathbb{R}_{\ge 0}}

\newcommand{\unk}{\mathord{?}}

\newcolumntype{P}[1]{>{\raggedright\arraybackslash}p{#1}}
\newcolumntype{Y}[1]{>{\centering\arraybackslash}p{#1}}
\newcolumntype{C}{>{\centering\arraybackslash}m{0.63in}}
\newcolumntype{R}{>{\raggedleft\arraybackslash}m{0.44in}@{\hspace{0.19in}}}
\newcolumntype{G}{>{\centering\arraybackslash}m{0.41in}}
\newcolumntype{N}{>{\centering\arraybackslash}m{0.30in}}
\newcolumntype{M}{>{\centering\arraybackslash}m{0.60in}}
\newcolumntype{H}{>{\centering\arraybackslash}m{0.37in}}
\newcolumntype{T}{>{\centering\arraybackslash}m{0.58in}}
\newcolumntype{U}{>{\raggedleft\arraybackslash}m{0.47in}@{\hspace{0.11in}}}

  {\hfill$\Box$\par}

%% file: Sections/0-Abstract.tex
Active model learning is a crucial technique for automatically constructing formal models of black-box timed systems.
However, existing learning approaches typically construct a single monolithic automaton for a timed system with multiple synchronous components, obscuring its component structure and hindering component-level analysis.
We propose a compositional learning approach that constructs a deterministic real-time automaton (RTA) for each component with access only to the target system as a whole.
Upon termination, the learned RTAs form a network that exhibits exactly the same timed behaviors as the target system.
Within the respective resource limits, our implementation terminates successfully on 213 of the 240 ($88.8\%$) randomly generated benchmark instances and all six adapted case-study benchmarks, and requires less learning time than monolithic approaches on some instances.

%% file: Sections/1-Introduction.tex
\section{Introduction}
\emph{Timed systems} are widely used in real-world applications, including embedded systems, automotive control, and communication protocols.
Model-based verification of such systems relies on formal models that faithfully capture their real-time behaviors.
However, building such models manually is labor-intensive and error-prone, especially when the system is accessible only as a black box.

\emph{Active model learning}~\cite{Model_learning}, which collects information from a target system on demand to construct its model automatically, has therefore become an important modeling technique.
The $\mathtt{L}^*$ algorithm~\cite{Angluin87} formalizes this interaction through a \emph{learner} that issues \emph{membership queries} and \emph{equivalence queries} to a \emph{teacher} with access to the target system.
A membership query asks whether the target system can exhibit a given behavior, whereas an equivalence query asks whether a hypothesis matches the target system.
By organizing the teacher's answers in an \emph{observation table}, the learner constructs a hypothesis and incrementally refines it until the teacher returns true for an equivalence query.
Subsequent research has applied active model learning to various model classes~\cite{I_O,register1,mealy} and developed more efficient learning algorithms~\cite{TTT,Lsharp,Lsharp2}.
For timed systems, active learning algorithms have been developed for real-time automata (RTAs)~\cite{RTA_learning}, Mealy machines with timers~\cite{mealy_multi_timers,vaandrager2023learning}, and timed automata~\cite{DOTA_learning,HSCC24_Multi_Clock,CAV2023_Multi_Clock}.
Tree-based information structures have also been introduced to improve the efficiency of learning timed systems~\cite{TTT_TA,Science_China_TTT,timed_L_sharp}.
However, these approaches construct a single monolithic automaton for the entire timed system.
For a timed system composed of multiple synchronous components, a monolithic model obscures the component structure, making it difficult to perform component-level analysis or apply verification techniques such as compositional abstraction~\cite{AGR_2023,Broadcast_2025,lin2013learning}.
Moreover, the global state space can grow rapidly with the number and complexity of components, increasing the potential cost of monolithic learning.

Motivated by these considerations, this paper presents a compositional active learning approach for timed synchronous systems.
Given the number of system components and their respective action sets, the approach constructs a separate deterministic RTA for each component.
Upon termination, the learned RTAs form a network that exhibits exactly the same timed behaviors as the target system.
Building on compositional active learning for untimed synchronous systems~\cite{Compositional_DFA_learning}, our approach comprises multiple component learners, a query adapter, and a system-level teacher.
Each learner learns an individual component by issuing component-level queries, such as whether the component can exhibit a given behavior.
In contrast, the teacher can answer only system-level queries, such as whether the entire system can exhibit a given behavior.
The query adapter translates component-level queries into system-level queries and derives component-level answers from the teacher's responses.

Although this architecture has been successfully applied to untimed synchronous systems, extending it to timed synchronous systems requires addressing several challenges.
First, continuous time gives rise to infinitely many time delays, so information collected from individual timed words must be generalized.
The adapter abstracts timed words as \emph{RTA region words} and represents confirmed component behaviors as \emph{positive-evidence RTAs}.
Second, answering a component-level membership query may require lifting the queried component behavior to a system-level behavior that respects the timing and synchronization constraints of all components.
The adapter uses timed reachability over these positive-evidence RTAs to find such a behavior.
Finally, with only system-level access, the adapter may temporarily be unable to answer some component-level membership queries and may return \emph{spurious counterexamples} to component-level equivalence queries.
Each component learner runs an $\mathtt{L}^*$-style procedure that organizes incomplete timed information in a \emph{timed observation table with unknown entries} (TOTU), uses SMT to construct RTA hypotheses, and employs a backtracking mechanism to handle possible spurious counterexamples.

In summary, the contributions of this paper are as follows.
\vspace{-1mm}
\begin{itemize}[topsep=2pt,itemsep=1pt,parsep=1pt]
  \item 
  We present a compositional learning approach that learns an RTA for each component of a timed synchronous system using only system-level access.
  \item 
  We develop a query adapter equipped with positive-evidence RTAs and TOTU-based component learners to address the challenges encountered in compositional learning of timed synchronous systems.
  \item
  We implement and evaluate our approach, which terminates successfully on 213 of the 240 ($88.8\%$) randomly generated benchmark instances and all six adapted case-study benchmarks within the respective resource limits, and requires less learning time than monolithic approaches on some instances.
\end{itemize}

%% file: Sections/2-Preliminaries.tex
\section{Preliminaries}
\label{Sec:Preliminaries}

Let $\mathbb{R}_{\geq 0}$ and $\mathbb{N}$ be the sets of non-negative reals and natural numbers, respectively. $\mathbb{B}=\{\top,\bot\}$ denotes the Boolean set, where $\top$ is true and $\bot$ is false.
Let $\Sigma$ be a finite action set, called the \emph{alphabet}.
A \emph{timed word} over $\Sigma$ is a finite sequence $w = (d_1,\sigma_1)(d_2,\sigma_2)\cdots(d_n,\sigma_n) \in (\NR \times \Sigma)^*$.
For $1\le i\le n$, $d_i$ is the delay before action $\sigma_i$ occurs, measured from the initial instant if $i=1$ and from the occurrence of $\sigma_{i-1}$ if $i>1$.
Each pair $(d_i,\sigma_i)$ is called a \emph{timed action}, and $|w|$ denotes the number of timed actions in $w$.
The empty word is denoted by $\epsilon$.

Let $\mathcal{X}$ be a set of real-valued variables, called \emph{clocks}.
A \emph{clock valuation} over $\mathcal{X}$ is a function $v:\mathcal{X}\to\NR$.
We write $\{x_1\mapsto t_1,\dots,x_n\mapsto t_n\}$ for the clock valuation that assigns value $t_i$ to clock $x_i$, for $1\le i\le n$.
All clocks grow at the same rate as time elapses.
For $d\in\NR$, we write $v+d$ for the valuation that maps each clock $x\in\mathcal{X}$ to $v(x)+d$.
We write $\Phi(\mathcal{X})$ for the set of constraints over $\mathcal{X}$, generated by $\phi ::= x \le \kappa \mid x < \kappa \mid \kappa \le x \mid \kappa < x \mid \phi_1 \wedge \phi_2$, where $x\in\mathcal{X}$ and $\kappa\in\mathbb{N}$.
We write $v\models\phi$ if constraint $\phi$ evaluates to true under $v$.
A \emph{real-time automaton} (RTA) is a finite automaton equipped with a single clock.

\begin{definition}[Real-Time Automaton]
\label{def:rta}
A \emph{real-time automaton} is a tuple $\mathcal{A} = (Q, q_0, \Sigma, x, \Delta)$, where $Q$ is a finite set of \emph{locations}, $q_0 \in Q$ is the initial location, $\Sigma$ is the alphabet, $x$ is the single clock, and $\Delta \subseteq Q \times \Sigma \times \Phi(\{x\}) \times Q$ is the set of transitions.
\end{definition}

Each transition $\delta = (q,\sigma,\phi,q')\in \Delta$ represents a jump from the source location $q$ to the target location $q'$ by performing an action $\sigma$ when its guard $\phi\in\Phi(\{x\})$ is satisfied by the current clock valuation.
After each transition, the clock $x$ is reset to 0.
A \emph{state} of $\mathcal{A}$ is a pair $(q,v)$, where $q \in Q$ is a location and $v$ is a clock valuation over $\{x\}$.
A \emph{run} of $\mathcal{A}$ on a timed word $w = (d_1,\sigma_1)(d_2,\sigma_2)\cdots(d_n,\sigma_n)$ is a sequence $ 
  \rho=(q_0,v_0) \xrightarrow{d_1,\sigma_1} (q_1,v_1)
  \xrightarrow{d_2,\sigma_2} \cdots
  \xrightarrow{d_n,\sigma_n} (q_n,v_n) $
such that $v_0(x)=0$ and, for every $1\leq j\leq n$, there exists a transition $(q_{j-1},\sigma_j,\phi_j,q_j)\in \Delta$ with $v_{j-1}+d_j\models \phi_j$, and $v_j(x)=0$.
We call $\mathcal{A}$ \emph{deterministic} if it has at most one run on any given timed word.

In this paper, we focus on the execution behaviors of timed systems and use RTAs as executability models without distinguishing \emph{accepting} and \emph{non-accepting} locations.
On this basis, the \emph{timed language} of $\mathcal{A}$ is defined as
\begin{equation*}
  \mathcal{L}(\mathcal{A}) = \bigl\{\, w \in (\NR \times \Sigma)^* \;\bigm|\; \mathcal{A} \text{ has a run on } w \,\bigr\}.
\end{equation*}
Thus, $\mathcal{L}(\mathcal{A})$ is prefix-closed, since every prefix of a run is again a run.
To model a system composed of multiple interacting components, RTAs can be composed into a \emph{network of real-time automata} (NRTA).

\begin{definition}[Network of Real-Time Automata]
\label{def:nrta}
A \emph{network of real-time automata} is a tuple $\mathcal{N} = (\mathcal{A}_1, \ldots, \mathcal{A}_n)$, where each component $\mathcal{A}_i = (Q_i, q_0^i, \Sigma_i,$ $x_i, \Delta_i)$ is an RTA and the clocks $x_1,\ldots,x_n$ are pairwise distinct. 
\end{definition}

For an NRTA $\mathcal{N}$, we instantiate the generic notation above as $\Sigma=\bigcup_{i=1}^n\Sigma_i$ for its alphabet and $\mathcal{X}=\{x_1,\ldots,x_n\}$ for its clock set.
For a valuation $v$ over $\mathcal{X}$ and $\sigma\in\Sigma$, $v[\sigma]$ is a new valuation over $\mathcal{X}$ with $v[\sigma](x_i)=0$ if $\sigma\in\Sigma_i$ and $v[\sigma](x_i)=v(x_i)$ otherwise, for every $1\le i\le n$.
We say that $\mathcal{N}$ has a run on a timed word $u=(d_1,\sigma_1)(d_2,\sigma_2)\cdots(d_m,\sigma_m)\in(\NR\times\Sigma)^*$ if there is a sequence
\begingroup
\setlength{\abovedisplayskip}{1.0mm}
\setlength{\abovedisplayshortskip}{1.0mm}
\setlength{\belowdisplayskip}{1.2mm}
\setlength{\belowdisplayshortskip}{1.2mm}
\[
  (\mathbf{q}_0,v_0)
  \xrightarrow{d_1,\sigma_1}
  (\mathbf{q}_1,v_1)
  \xrightarrow{d_2,\sigma_2}
  \cdots
  \xrightarrow{d_m,\sigma_m}
  (\mathbf{q}_m,v_m),
\]
\endgroup
\noindent
where each \emph{location vector} $\mathbf{q}_j=(q_j^1,\ldots,q_j^n)\in Q_1\times\cdots\times Q_n$, $\mathbf{q}_0=(q_0^1,\ldots,q_0^n)$, $v_0(x_i)=0$ for every $1\le i\le n$, and for every $1\le j\le m$:
\begin{enumerate}[topsep=5pt,itemsep=1pt,parsep=1pt]
  \item
		  For every component $\mathcal{A}_i$, there is a transition $(q_{j-1}^i,\sigma_j,\phi_j^i,q_j^i)\in\Delta_i$ satisfying $v_{j-1}+d_j\models\phi_j^i$ if $\sigma_j\in\Sigma_i$, and $q_j^i=q_{j-1}^i$ otherwise;
  \item
  $v_j = (v_{j-1}+d_j)[\sigma_j]$.
\end{enumerate}
\noindent
Thus, NRTA components follow \emph{multicast synchronization}: an action $\sigma\in\Sigma$ can occur only when all components with $\sigma\in\Sigma_i$ can take a $\sigma$-transition, and they take it simultaneously.
The \emph{timed language} of an NRTA $\mathcal{N}$ is
\begingroup
\setlength{\abovedisplayskip}{1.7mm}
\setlength{\abovedisplayshortskip}{1.7mm}
\setlength{\belowdisplayskip}{2mm}
\setlength{\belowdisplayshortskip}{2mm}
\[
  \mathcal{L}(\mathcal{N})
  =
  \bigl\{\, u\in(\NR\times\Sigma)^*
  \;\bigm|\;
  \mathcal{N}\text{ has a run on }u\,\bigr\}.
\]
\endgroup
Since every prefix of an NRTA run is again an NRTA run, $\mathcal{L}(\mathcal{N})$ is prefix-closed.

For a timed word $u\in(\NR\times\Sigma)^*$, the behavior observed by component $\mathcal{A}_i$ consists only of actions from its alphabet $\Sigma_i$.
This behavior is captured by the \emph{projection} $\project{u}{i}$, which keeps only actions in $\Sigma_i$ and sums the delays between consecutive retained actions.
For example, if $\Sigma_i=\{a,c\}$ and $u=(1,a)(2,b)(3,c)(4,b)(5,a)$, then $\project{u}{i}=(1,a)(5,c)(9,a)$, where the delay before $c$ is $2+3=5$ and the delay before the final $a$ is $4+5=9$.
With this projection, the language of an NRTA admits the following component-wise characterization.

\vspace{-1mm}
\begin{lemma}
\label{lemma:proj_sound}
For an NRTA $\mathcal{N}=(\mathcal{A}_1,\ldots,\mathcal{A}_n)$ and a timed word $u\in(\NR\times\Sigma)^*$,
\vspace{-3mm}
\[
  u\in\mathcal{L}(\mathcal{N})
  \quad \Longleftrightarrow \quad
  \forall\,1\le i\le n:\project{u}{i}\in\mathcal{L}(\mathcal{A}_i).
\]
\end{lemma}
This lemma follows from the definitions of NRTA languages and projections.
The proof is provided in Appendix~\ref{app:proof-projection}.

%% file: Sections/3-Problem_Formulation.tex
\section{Learning Setting and Approach Overview}
\label{Sec:Problem}

\subsection{Learning Setting}
\label{Sec:Setting}
Consider a multi-component black-box timed system represented by an NRTA $\mathcal{N} = (\mathcal{A}_1,\ldots,\mathcal{A}_n)$, where each $\mathcal{A}_i$ is a deterministic RTA.
The goal is to automatically construct an NRTA $\mathcal{N}'=(\mathcal{H}_1,\ldots,\mathcal{H}_n)$ using \emph{limited prior information} and a \emph{system-level teacher}, such that $\mathcal{L}(\mathcal{N}')=\mathcal{L}(\mathcal{N})$.

\vspace{1ex}
\noindent
\textbf{The prior information} consists of the number of components $n$ and the component alphabets $\Sigma_1,\ldots,\Sigma_n$, as assumed in compositional learning of untimed synchronous systems~\cite{Compositional_DFA_learning}.
Architecture descriptions commonly provide such interface-level information by identifying modules and synchronization actions.
The components' internal locations and transitions remain unavailable.

\vspace{1ex}
\noindent
\textbf{The system-level teacher} holds the underlying target NRTA $\mathcal{N}$ and is able to answer two kinds of queries:

\begin{itemize}[leftmargin=*,topsep=2pt,itemsep=1pt,parsep=1pt]
  \item
  \textbf{Membership query:}
  given a timed word $u$, it returns whether $u \in \mathcal{L}(\mathcal{N})$.

  \item
  \textbf{Equivalence query:}
  given a hypothesis NRTA $\mathcal{N}' = (\mathcal{H}_1,\ldots,\mathcal{H}_n)$, it returns $\top$ if $\mathcal{L}(\mathcal{N}')=\mathcal{L}(\mathcal{N})$, and otherwise returns $(\bot,u_{\mathit{ce}})$, where $u_{\mathit{ce}}\in(\mathcal{L}(\mathcal{N}')\setminus\mathcal{L}(\mathcal{N}))\cup(\mathcal{L}(\mathcal{N})\setminus\mathcal{L}(\mathcal{N}'))$ is a counterexample with minimal $|u_{\mathit{ce}}|$.
\end{itemize}

\noindent
This teacher setting fits integrated timed systems whose components cannot be queried separately, such as cyber-physical systems or software modules exposed only through a system interface.
As in prior active automata learning~\cite{DOTA_learning,RTA_learning,Angluin87}, the system-level teacher considered in this paper is idealized, with access to the target model $\mathcal{N}$ and exact answers to both query types.
All subsequent development assumes this idealized teacher.
In practical applications, membership queries can be controlled tests, whereas equivalence queries can be approximated by PAC-style testing or mutation testing~\cite{shen2020pac,TangSZAZZ22}.

Notably, the learned network $\mathcal{N}'$ need only satisfy $\mathcal{L}(\mathcal{N}')=\mathcal{L}(\mathcal{N})$, rather than $\mathcal{L}(\mathcal{H}_i)=\mathcal{L}(\mathcal{A}_i)$ for every $i$.
This is because some component behaviors cannot be identified from system-level queries alone, as illustrated in Section~\ref{Sec:Algorithm}.

\subsection{Approach Overview}
\label{Sec:ApproachOverview}

\begin{figure}[t]
  \centering
  \input{Figures/architecture_tikz}
  \vspace{-5mm}
  \caption{Learning architecture for compositional learning of $\mathcal{N}$.}
  \vspace{-4mm}
  \label{fig:architecture}
\end{figure}
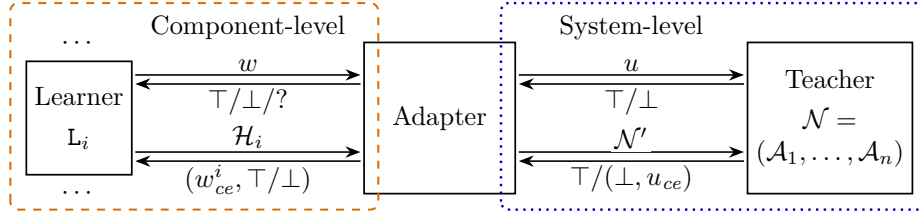

Figure~\ref{fig:architecture} shows our compositional learning architecture, which consists of a system-level teacher, a query adapter, and $n$ component learners.
Each learner $\mathtt{L}_i$ runs an $\mathtt{L}^*$-style active learning algorithm to learn a hypothesis $\mathcal{H}_i$ for the hidden component $\mathcal{A}_i$ by issuing component-level membership and equivalence queries to the adapter.
The adapter mediates between the learners and the system-level teacher by translating component-level queries into system-level queries and mapping returned counterexamples back to component-level refinements.

To answer a component-level membership query asking whether a timed word $w\in(\NR\times\Sigma_i)^*$ belongs to $\mathcal{L}(\mathcal{A}_i)$, the adapter builds a timed word $u\in(\NR\times\Sigma)^*$ from $w$ and asks the system-level teacher whether $u\in\mathcal{L}(\mathcal{N})$.
If the teacher's answer determines the component-level result, the adapter returns $\top$ or $\bot$; otherwise, it returns $\unk$, meaning unknown with the current information.

Once a learner constructs a hypothesis $\mathcal{H}_i$ from the collected information, it submits a component-level equivalence query to the adapter on whether $\mathcal{L}(\mathcal{H}_i)=\mathcal{L}(\mathcal{A}_i)$.
The adapter does not answer this query immediately.
Instead, it waits until all component hypotheses are available, assembles a hypothesis for the entire system $\mathcal{N}'=(\mathcal{H}_1,\ldots,\mathcal{H}_n)$, and submits a system-level equivalence query to the teacher.
If the teacher returns $\top$, learning terminates.
Otherwise, from the returned system-level counterexample $u_{\mathit{ce}}$, the adapter derives a component-level answer $(w_{\mathit{ce}}^i,\beta)$ for each learner $\mathtt{L}_i$, where $w_{\mathit{ce}}^i\in(\NR\times\Sigma_i)^*$ and $\beta\in\{\top,\bot\}$ indicates whether the learner should include $w_{\mathit{ce}}^i$ in ($\top$) or exclude it from ($\bot$) the language of its next hypothesis.

%% file: Figures/architecture_tikz.tex
\begin{tikzpicture}[
  x=1cm,
  y=1cm,
  node distance=3.05cm,
  >=Stealth,
  font=\normalsize,
  learner/.style={draw, line width=0.65pt, minimum width=1.3cm,
    minimum height=1.5cm, align=center, fill=white},
  adapter/.style={draw, line width=0.65pt, minimum width=2.00cm,
    minimum height=2.0cm, align=center, fill=white},
  teacher/.style={draw, line width=0.65pt, minimum width=2.00cm,
    minimum height=2.0cm, align=center, fill=white},
  flow/.style={->, line width=0.6pt, shorten >=1pt, shorten <=1pt},
  message/.style={font=\normalsize, inner sep=1.2pt, fill=white}
]
  \node[learner] (learner) at (0,0) {Learner\\[3pt]$\mathtt{L}_i$};
  \node[above=1.2pt of learner.north] {$\cdots$};
  \node[below=1.2pt of learner.south] {$\cdots$};
  \node[adapter,right=of learner] (adapter) {Adapter};
  \node[teacher,right=of adapter] (teacher) {Teacher\\[3pt]$\mathcal{N}=$\\$(\mathcal{A}_1,\ldots,\mathcal{A}_n)$};

  \path ([yshift=1.2cm]learner.east) --
    node[midway] (componentlevel) {Component-level} ([yshift=1.2cm]adapter.west);
  \path ([yshift=1.2cm]adapter.east) --
    node[midway] (systemlevel) {System-level} ([yshift=1.2cm]teacher.west);

  \draw[flow] ([yshift=0.57cm]learner.east) -- node[message,above] {$w$} ([yshift=0.57cm]adapter.west);
  \draw[flow] ([yshift=0.45cm]adapter.west) -- node[message,below] {$\top/\bot/\unk$} ([yshift=0.45cm]learner.east);

  \draw[flow] ([yshift=-0.45cm]learner.east) -- node[message,above] {$\mathcal{H}_i$} ([yshift=-0.45cm]adapter.west);
  \draw[flow] ([yshift=-0.57cm]adapter.west) -- node[message,below] {$(w_{\mathit{ce}}^i,\top/\bot)$} ([yshift=-0.57cm]learner.east);

  \draw[flow] ([yshift=0.57cm]adapter.east) -- node[message,above] {$u$} ([yshift=0.57cm]teacher.west);
  \draw[flow] ([yshift=0.45cm]teacher.west) -- node[message,below] {$\top/\bot$} ([yshift=0.45cm]adapter.east);

  \draw[flow] ([yshift=-0.45cm]adapter.east) -- node[message,above] {$\mathcal{N}'$} ([yshift=-0.45cm]teacher.west);
  \draw[flow] ([yshift=-0.57cm]teacher.west) -- node[message,below] {$\top/(\bot,u_{\mathit{ce}})$} ([yshift=-0.57cm]adapter.east);

  \def\levelpadding{0.20cm}
  \draw[orange!85!black, dashed, line width=0.75pt, rounded corners=3pt]
    ($(learner.west |- adapter.south)+(-\levelpadding,-\levelpadding)$)
    rectangle
    ($(adapter.west |- componentlevel.north)+(\levelpadding,0.05cm)$);
  \draw[blue!75!black, dotted, line width=0.9pt, rounded corners=3pt]
    ($(adapter.east |- adapter.south)+(-\levelpadding,-\levelpadding)$)
    rectangle
    ($(teacher.east |- systemlevel.north)+(\levelpadding,0.05cm)$);
\end{tikzpicture}%

%% file: Sections/4-Query_Adapter.tex
\section{Query Adapter}
\label{Sec:Adapter}
For brevity, we write $\mathsf{MQ}(u)$ for the teacher's answer to the system-level membership query on whether $u\in\mathcal{L}(\mathcal{N})$, and $\mathsf{EQ}(\mathcal{N}')$ for its answer to the system-level equivalence query on whether $\mathcal{L}(\mathcal{N}')=\mathcal{L}(\mathcal{N})$.

\subsection{Positive Evidence}
\label{Sec:PositiveEvidence}
We call a timed word $w\in(\NR\times\Sigma_i)^*$ \emph{positive evidence} for component $\mathcal{A}_i$ if the adapter has established that $w\in\mathcal{L}(\mathcal{A}_i)$.

\vspace{-0.5mm}
\begin{lemma}
\label{lemma:partner_infer}
For a timed word $u\in(\NR\times\Sigma)^*$ and a component $\mathcal{A}_i$ of $\mathcal{N}$, if $\project{u}{j}\in\mathcal{L}(\mathcal{A}_j)$ for every $j\neq i$, then
\vspace{-2mm}
\[
  \project{u}{i}\in\mathcal{L}(\mathcal{A}_i)
  \quad\Longleftrightarrow\quad
  u\in\mathcal{L}(\mathcal{N}).
\]
\end{lemma}
\begin{proof}
By Lemma~\ref{lemma:proj_sound}, $u\in\mathcal{L}(\mathcal{N})$ iff all its component projections belong to their respective component languages.
The premise of Lemma~\ref{lemma:partner_infer} already establishes projection membership for every $j\neq i$, leaving only $\project{u}{i}\in\mathcal{L}(\mathcal{A}_i)$, which proves the stated equivalence.\qed
\end{proof}

Based on Lemma~\ref{lemma:partner_infer}, the adapter incrementally accumulates positive evidence for each component to help derive answers to component-level queries from the system-level teacher's responses.
Because time is continuous, the adapter generalizes positive evidence beyond isolated timed words using \emph{RTA region words}.

\vspace{-0.5mm}
\begin{definition}[RTA Region Word]
\label{def:rta_region_word}
Let $w=(d_1,\sigma_1)\cdots(d_m,\sigma_m)$ be a timed word.
Its \emph{RTA region word} is $\rabs{w}=(r_1,\sigma_1)\cdots(r_m,\sigma_m)$, where, for each $1\le j\le m$, $r_j=[d_j,d_j]$ if $d_j\in\mathbb{N}$ and $r_j=(\lfloor d_j\rfloor,\lceil d_j\rceil)$ otherwise\footnote{Unlike standard regions, which abstract clock valuations~\cite{AlurD94}, these intervals abstract delays and can be viewed as a variant of region abstraction tailored to our setting.}.
\end{definition}
\noindent
For instance, if $w=(0,a)(2.4,b)(5,c)$, then $\rabs{w}=([0,0],a)((2,3),b)([5,5],c)$.
Each pair $(r_j,\sigma_j)$ in $\rabs{w}$ is called an \emph{RTA region action}.

\begin{lemma}
\label{lemma:rta_region_equal}
  For any timed words $w,w'\in(\NR\times\Sigma_i)^*$, if $\rabs{w}=\rabs{w'}$, then
  \vspace{-2mm}
  \[
    w\in\mathcal{L}(\mathcal{A}_i)
    \quad\Longleftrightarrow\quad
    w'\in\mathcal{L}(\mathcal{A}_i).
  \]
\end{lemma}
\noindent
This follows because RTA transition guards use only integer constants and the single clock is reset after every transition.
The proof is detailed in Appendix~\ref{app:proof-region-equality}.
Thus, once a timed word $w$ is confirmed as positive evidence for component $\mathcal{A}_i$, any timed word $w'$ with $\rabs{w'}=\rabs{w}$ is also positive evidence for $\mathcal{A}_i$.

For each component $\mathcal{A}_i$, the adapter maintains a \emph{positive-evidence RTA} $\mathcal{P}_i=(Q_i^+,q_0^{i,+},\Sigma_i,x_i,\Delta_i^+)$ such that every $w\in\mathcal{L}(\mathcal{P}_i)$ is positive evidence for $\mathcal{A}_i$.
Initially, $\mathcal{P}_i$ contains only the initial location $q_0^{i,+}$ and no transitions, i.e., $Q_i^+=\{q_0^{i,+}\}$, $\Delta_i^+=\emptyset$ and $\mathcal{L}(\mathcal{P}_i) = \{\epsilon\}$.
When a timed word $w$ is confirmed as positive evidence for $\mathcal{A}_i$, the adapter updates $\mathcal{P}_i$ according to Algorithm~\ref{alg:UpdatePositiveEvidence}.
The algorithm starts from the initial location of $\mathcal{P}_i$ and scans the timed actions of $w$ (Lines~\ref{lines:addpe_init}--\ref{lines:addpe_loop}).
At a location $q$, if there exists a transition $(q,\sigma_j,\phi,q')\in\Delta_i^+$ such that the clock valuation $\{x_i\mapsto d_j\}$ satisfies the guard $\phi$, the algorithm follows it and moves to $q'$ (Lines~\ref{lines:addpe_reuse}--\ref{lines:addpe_follow}).
If no such transition exists, it creates a fresh successor and inserts $(q,\sigma_j,\mathsf{guard}(x_i,d_j),q')$ (Lines~\ref{lines:addpe_fresh}--\ref{lines:addpe_insert}), where $\mathsf{guard}(x_i,d_j)$ is $d_j\le x_i\wedge x_i\le d_j$ if $d_j\in\mathbb{N}$, and $\lfloor d_j\rfloor<x_i\wedge x_i<\lceil d_j\rceil$ otherwise.
By construction, the followed or inserted transition sequence corresponds to the RTA region word $\rabs{w}$.
Hence, after the update, every timed word $w'$ with $\rabs{w'}=\rabs{w}$ belongs to $\mathcal{L}(\mathcal{P}_i)$.
This update procedure is denoted by $\mathsf{addPE}_i(w)$.
Since $\mathsf{addPE}_i$ is the only procedure that updates $\mathcal{P}_i$, its construction and Lemma~\ref{lemma:rta_region_equal} ensure that the language inclusion $\mathcal{L}(\mathcal{P}_i)\subseteq\mathcal{L}(\mathcal{A}_i)$ is preserved throughout learning.

\input{Algorithms/UpdatePositiveEvidence}

\subsection{Membership Query}
\label{Sec:MQPipeline}

Given a component-level membership query from $\mathtt{L}_i$ on whether a timed word $w\in(\NR\times\Sigma_i)^*$ belongs to $\mathcal{L}(\mathcal{A}_i)$, the adapter answers it in three stages.

\smallskip
\noindent\textbf{Stage~1.}
The adapter first directly views $w$ as a system-level timed word over $\Sigma$ and submits it to the teacher.
By Lemma~\ref{lemma:proj_sound}, if $\mathsf{MQ}(w)=\top$, then $\project{w}{j}\in\mathcal{L}(\mathcal{A}_j)$ for every $1\le j\le n$.
In this case, the adapter calls $\mathsf{addPE}_j(\project{w}{j})$ for each $j$ and returns $\top$ to $\mathtt{L}_i$.
Conversely, if $\mathsf{MQ}(w)=\bot$, then there must be some component $\mathcal{A}_j$ with $\project{w}{j}\notin\mathcal{L}(\mathcal{A}_j)$.
In this case, the adapter proceeds to Stage~2.

\smallskip
\noindent\textbf{Stage~2.}
The adapter checks whether $\project{w}{j}\in\mathcal{L}(\mathcal{P}_j)$ for every $j\neq i$.
If so, the language inclusion $\mathcal{L}(\mathcal{P}_j)\subseteq\mathcal{L}(\mathcal{A}_j)$ implies that $\project{w}{j}\in\mathcal{L}(\mathcal{A}_j)$ for every $j\neq i$.
Since $\project{w}{i}=w$, Lemma~\ref{lemma:partner_infer} implies that $w\notin\mathcal{L}(\mathcal{A}_i)$, and the adapter returns $\bot$.
If some index $j\neq i$ satisfies $\project{w}{j}\notin\mathcal{L}(\mathcal{P}_j)$, the adapter cannot yet determine whether $w\in\mathcal{L}(\mathcal{A}_i)$ and proceeds to Stage~3 for further analysis.

\smallskip
\noindent\textbf{Stage~3.}
The adapter attempts to construct a system-level timed word $u$ such that $\rabs{\project{u}{i}}=\rabs{w}$ and $\project{u}{j}\in\mathcal{L}(\mathcal{P}_j)$ for every $j\neq i$.
Given $w=(d_1,\sigma_1)\cdots(d_m,\sigma_m)$, the adapter first builds an RTA $\mathcal{W}=(Q^w,q_0^w,\Sigma_i,x_i,\Delta^w)$, where $Q^w$ consists of $m+1$ distinct locations $q_0^w,\ldots,q_m^w$, and $\Delta^w=\{(q_{j-1}^w,\sigma_j,\mathsf{guard}(x_i,d_j),q_j^w)\mid 1\le j\le m\}$.
Thus, by construction, for any timed word $w'\in(\NR\times\Sigma_i)^*$ with $|w'|=|w|$, $w'\in\mathcal{L}(\mathcal{W})$ iff $\rabs{w'}=\rabs{w}$.
Then, the adapter constructs an NRTA
\begingroup
\setlength{\abovedisplayskip}{1.7mm}
\setlength{\abovedisplayshortskip}{1.7mm}
\setlength{\belowdisplayskip}{2mm}
\setlength{\belowdisplayshortskip}{2mm}
\[
  \mathcal{N}_i^w
  =
  (\mathcal{P}_1,\ldots,\mathcal{P}_{i-1},\mathcal{W},\mathcal{P}_{i+1},\ldots,\mathcal{P}_n)
\]
\endgroup
and checks whether $\mathcal{N}_i^w$ has a run reaching a location vector whose $i$-th entry is $q_m^w$.
If such a run is found, it induces the required timed word $u$.
The adapter then issues a system-level membership query to the teacher on whether $u\in\mathcal{L}(\mathcal{N})$, obtaining the answer $\mathsf{MQ}(u)$.
Since $\mathcal{L}(\mathcal{P}_j)\subseteq\mathcal{L}(\mathcal{A}_j)$ gives $\project{u}{j}\in\mathcal{L}(\mathcal{A}_j)$ for every $j\neq i$, Lemmas~\ref{lemma:partner_infer} and~\ref{lemma:rta_region_equal} imply that $\mathsf{MQ}(u)$ is also the answer to the original component-level membership query on whether $w\in \mathcal{L}(\mathcal{A}_i)$.
The adapter returns $\mathsf{MQ}(u)$ to $\mathtt{L}_i$ and calls $\mathsf{addPE}_i(w)$ when $\mathsf{MQ}(u)=\top$.
If no such run exists, the adapter cannot yet determine whether $w\in\mathcal{L}(\mathcal{A}_i)$ and returns $\unk$.

\subsection{Equivalence Query}
\label{Sec:CEClassification}
Once learner $\mathtt{L}_i$ constructs a hypothesis RTA $\mathcal{H}_i$, it submits a component-level equivalence query to the adapter, asking whether $\mathcal{L}(\mathcal{H}_i)=\mathcal{L}(\mathcal{A}_i)$.
The adapter waits until every learner has issued an equivalence query, collects their hypotheses, assembles $\mathcal{N}'=(\mathcal{H}_1,\ldots,\mathcal{H}_n)$, and asks the teacher whether $\mathcal{L}(\mathcal{N}')=\mathcal{L}(\mathcal{N})$.
If $\mathsf{EQ}(\mathcal{N}')=\top$, learning terminates.
Otherwise, the adapter receives a counterexample $u_{\mathit{ce}}$.
For each component $\mathcal{A}_i$, write $w_{\mathit{ce}}^i=\project{u_{\mathit{ce}}}{i}$.
The adapter returns component-level refinements by considering the following two cases.

\smallskip
\noindent\textbf{Case 1: $u_{\mathit{ce}}\in\mathcal{L}(\mathcal{N})\setminus\mathcal{L}(\mathcal{N}')$.}
In this case, $u_{\mathit{ce}}$ is a \emph{positive system-level counterexample}.
By Lemma~\ref{lemma:proj_sound}, $w_{\mathit{ce}}^i\in\mathcal{L}(\mathcal{A}_i)$ for every component $i$.
Therefore, the adapter calls $\mathsf{addPE}_i(w_{\mathit{ce}}^i)$ and returns $(w_{\mathit{ce}}^i,\top)$ to learner $\mathtt{L}_i$, indicating that the learner should ensure $w_{\mathit{ce}}^i\in\mathcal{L}(\mathcal{H}_i)$ in its next hypothesis.

\smallskip
\noindent\textbf{Case 2: $u_{\mathit{ce}}\in\mathcal{L}(\mathcal{N}')\setminus\mathcal{L}(\mathcal{N})$.}
In this case, $u_{\mathit{ce}}$ is a \emph{negative system-level counterexample}.
Let $p$ be the timed word obtained from $u_{\mathit{ce}}=(d_1,\sigma_1)\cdots(d_m,\sigma_m)$ by removing its last timed action.
Since $\mathcal{L}(\mathcal{N}')$ is prefix-closed, $p\in\mathcal{L}(\mathcal{N}')$.
Because the teacher returns a shortest counterexample, $p$ cannot be a shorter system-level counterexample, and hence $p\in\mathcal{L}(\mathcal{N})$.
By Lemma~\ref{lemma:proj_sound}, $\project{p}{i}\in\mathcal{L}(\mathcal{A}_i)$ for every component $i$, so the adapter calls $\mathsf{addPE}_i(\project{p}{i})$.
The adapter then returns a component-level refinement to each learner $\mathtt{L}_i$ according to the following rules.
\begin{enumerate}[topsep=3pt,itemsep=1pt,parsep=1pt]
  \item 
  If $\sigma_m\notin\Sigma_i$, then $w_{\mathit{ce}}^i=\project{p}{i}\in\mathcal{L}(\mathcal{A}_i)$, and the adapter returns $(w_{\mathit{ce}}^i,\top)$;
  \item 
  If $\sigma_m\in\Sigma_i$ and $w_{\mathit{ce}}^i\in\mathcal{L}(\mathcal{P}_i)$, the adapter also returns $(w_{\mathit{ce}}^i,\top)$;
  \item 
  Otherwise, the adapter returns $(w_{\mathit{ce}}^i,\bot)$.
\end{enumerate}
By Lemma~\ref{lemma:proj_sound} and the language inclusion $\mathcal{L}(\mathcal{P}_i)\subseteq\mathcal{L}(\mathcal{A}_i)$, at least one word returned by the third rule is indeed outside the corresponding component language.
Notably, when the third rule applies to several learners, a returned component-level counterexample may be \emph{spurious}.
For some such learner $\mathtt{L}_i$, $w_{\mathit{ce}}^i$ may in fact belong to $\mathcal{L}(\mathcal{A}_i)$, although the collected evidence is insufficient for the adapter to infer this.
Such spurious counterexamples may mislead component learners, and Section~\ref{Sec:LearnerCounterexample} describes how they are handled.

%% file: Algorithms/UpdatePositiveEvidence.tex
\begin{compactalgorithm}[t]{-5mm}
  \caption{Updating $\mathcal{P}_i$ with positive evidence}
  \label{alg:UpdatePositiveEvidence}
  \SetKwInOut{Input}{input}
  \Input{the RTA $\mathcal{P}_i=(Q_i^+,q_0^{i,+},\Sigma_i,x_i,\Delta_i^+)$; \\
	  positive evidence $w=(d_1,\sigma_1)\cdots(d_m,\sigma_m)$ for $\mathcal{A}_i$.}
  $q \leftarrow q_0^{i,+}$\;
  \label{lines:addpe_init}
  \For{$j \gets 1$ \KwTo $m$}{
    \label{lines:addpe_loop}
    \eIf{$\exists (q,\sigma_j,\phi,q')\in\Delta_i^+ : \{x_i\mapsto d_j\}\models \phi$}{
      \label{lines:addpe_reuse}
      $q \leftarrow q'$\;
      \label{lines:addpe_follow}
    }{
      introduce a fresh location $q'$\;
      \label{lines:addpe_fresh}
      $Q_i^+ \leftarrow Q_i^+\cup\{q'\}$\;
      \label{lines:addpe_addloc}
      $\Delta_i^+ \leftarrow \Delta_i^+\cup\{(q,\sigma_j,\mathsf{guard}(x_i, d_j),q')\}$\;
      \label{lines:addpe_insert}
      $q \leftarrow q'$\;
      \label{lines:addpe_movefresh}
    }
  }
\end{compactalgorithm}

%% file: Sections/5-Component_Learner.tex
\section{Component Learner}
\label{Sec:Learner}

Each component learner $\mathtt{L}_i$ runs an $\mathtt{L}^*$-style learning algorithm using a \emph{timed observation table with unknown entries} (TOTU) to organize observed component behaviors.
To describe TOTU operations, we use $\cdot$ to denote the concatenation of RTA region words.
For learner $\mathtt{L}_i$, we write $\mathsf{MQ}_i(w)$ for the adapter's answer to its component-level membership query on timed word $w$.

\subsection{Timed Observation Table with Unknown Entries}
\label{Sec:Table}

\begin{definition}[TOTU]
\label{def:totu}
A \emph{timed observation table with unknown entries} for component $\mathcal{A}_i$ is a tuple $\mathbf{T} = (\Sigma_i, \bm{S}, \bm{R}, \bm{E}, f)$.
Its components are as follows.
\begin{itemize}[topsep=2pt,itemsep=1pt,parsep=1pt]
  \item $\Sigma_i$ is the alphabet of $\mathcal{A}_i$;
  \item $\bm{S}$, $\bm{R}$, and $\bm{E}$ are finite sets of RTA region words over $\Sigma_i$, called the representative, boundary, and suffix sets, respectively;
  \item $f : (\bm{S} \cup \bm{R}) \times \bm{E} \to \{\top, \bot, \unk\}$ is a three-valued classification function, where $f(\pi,e)$ records whether a timed word $w$ with $\rabs{w}=\pi\cdot e$ is known to belong to $\mathcal{L}(\mathcal{A}_i)$ ($\top$), known to be outside it ($\bot$), or not yet determined ($\unk$).
\end{itemize}
\end{definition}

Each pair $(\pi,e)\in(\bm{S}\cup\bm{R})\times\bm{E}$ indexes one table entry, whose value is $f(\pi,e)$.
Each RTA region word $\pi \in \bm{S} \cup \bm{R}$ indexes one row, whose value is $\mathit{row}(\pi) = (f(\pi, e))_{e \in \bm{E}}$.
Since some entries may have value $\unk$, we define the following \emph{potential equality} relation between rows.

\begin{definition}[Potential Equality]
\label{def:weak_eq}
Given a TOTU $\mathbf{T} = (\Sigma_i, \bm{S}, \bm{R}, \bm{E}, f)$, the rows indexed by $\pi,\pi'\in\bm{S}\cup\bm{R}$ are \emph{potentially equal}, denoted by $\mathit{row}(\pi) \approx \mathit{row}(\pi')$, if, for every $e \in \bm{E}$, either $f(\pi,e)=f(\pi',e)$, $f(\pi,e)=\unk$, or $f(\pi',e)=\unk$.
\end{definition}

\noindent
The value $\unk$ acts as a wildcard that may later match either $\top$ or $\bot$.
If $\mathit{row}(\pi) \not\approx \mathit{row}(\pi')$, then the two indexed rows are already distinguishable in the current table.
The learner maintains the following four standard table properties.
\begin{itemize}[topsep=2pt,itemsep=1pt,parsep=1pt]
  \item \emph{Disjoint:} $\bm{S}\cap\bm{R}=\emptyset$;
  \item \emph{Reduced:} $\forall \pi,\pi'\in\bm{S}$: $\pi\neq\pi' \Rightarrow \mathit{row}(\pi)\not\approx\mathit{row}(\pi')$;
  \item \emph{Closed:} $\forall \pi\in\bm{R}$, $\exists \pi'\in\bm{S}$: $\mathit{row}(\pi)\approx\mathit{row}(\pi')$;
  \item \emph{Consistent:} $\forall \pi,\pi'\in\bm{S}\cup\bm{R}$: if $\mathit{row}(\pi)\approx\mathit{row}(\pi')$ and there is an RTA region action $(r,\sigma)$ such that $\pi\cdot(r,\sigma),\pi'\cdot(r,\sigma)\in\bm{S}\cup\bm{R}$, then $\mathit{row}(\pi\cdot(r,\sigma))\approx\mathit{row}(\pi'\cdot(r,\sigma))$.
\end{itemize}
Intuitively, rows indexed by words in $\bm{S}$ represent potential hypothesis locations, while $\bm{R}$ stores additional words used to discover locations and transitions.
Disjointness keeps $\bm{S}$ and $\bm{R}$ separate, reducedness distinguishes rows indexed by $\bm{S}$, closedness matches each $\bm{R}$-indexed row to an $\bm{S}$-indexed row, and consistency preserves potential equality under one-step extensions.
A table is \emph{prepared} when it satisfies all four properties.

Learner $\mathtt{L}_i$ initializes $\mathbf{T}$ with $\bm{S}=\bm{E}=\{\epsilon\}$ and $\bm{R}=\{([0,0],\sigma)\mid\sigma\in\Sigma_i\}$.
The initial classification satisfies $f(\epsilon,\epsilon)=\top$, while every other entry has value $\unk$.
Then $\mathtt{L}_i$ calls $\mathsf{Fill}(\mathbf{T})$, which processes each $(\pi,e)$ with $f(\pi,e)=\unk$ through up to three stages, stopping once a stage assigns $f(\pi,e)$.

\smallskip
\noindent\textbf{Stage~1.}
Infer $f(\pi,e)$ using the prefix-closure of $\mathcal{L}(\mathcal{A}_i)$.
If there is some $(\pi',e')$ with $f(\pi',e')=\bot$ and $\pi'\cdot e'$ is a prefix of $\pi\cdot e$, set $f(\pi,e)=\bot$.
If there is some $(\pi',e')$ with $f(\pi',e')=\top$ and $\pi\cdot e$ is a prefix of $\pi'\cdot e'$, set $f(\pi,e)=\top$.

\smallskip
\noindent\textbf{Stage~2.}
Choose any timed word $w$ with $\rabs{w}=\pi\cdot e$\footnote{Such a timed word can be obtained by choosing one delay from each interval in the RTA region word, e.g., $3$ from $[3,3]$ and $2.5$ from $(2,3)$.} and check whether $w\in\mathcal{L}(\mathcal{P}_i)$.
If so, set $f(\pi,e)=\top$, since $\mathcal{L}(\mathcal{P}_i)\subseteq \mathcal{L}(\mathcal{A}_i)$.

\smallskip
\noindent\textbf{Stage~3.}
Issue a component-level membership query to the adapter on whether $w\in\mathcal{L}(\mathcal{A}_i)$, and set $f(\pi,e)=\mathsf{MQ}_i(w)$.

\smallskip
The disjoint and reduced properties are preserved by all TOTU operations introduced later in this paper.
If $\mathbf{T}$ is not closed or not consistent, the learner $\mathtt{L}_i$ adjusts it by the following two operations, respectively.

\smallskip
\noindent\textbf{Making $\mathbf{T}$ closed.}
If there exists $\pi\in\bm{R}$ such that $\mathit{row}(\pi)\not\approx\mathit{row}(\pi')$ for all $\pi'\in\bm{S}$, $\mathtt{L}_i$ moves $\pi$ from $\bm{R}$ to $\bm{S}$.
For each $\sigma\in\Sigma_i$, let $\pi''=\pi\cdot([0,0],\sigma)$.
If $\pi''\notin\bm{S}\cup\bm{R}$, $\mathtt{L}_i$ adds $\pi''$ to $\bm{R}$ and sets $f(\pi'',e)=\unk$ for all $e\in\bm{E}$.

\smallskip
\noindent\textbf{Making $\mathbf{T}$ consistent.}
In this case, there exist $\pi,\pi'\in\bm{S}\cup\bm{R}$ with $\mathit{row}(\pi)\approx\mathit{row}(\pi')$ and an RTA region action $(r,\sigma)$ such that $\pi\cdot(r,\sigma),\pi'\cdot(r,\sigma)\in\bm{S}\cup\bm{R}$.
For some $e\in\bm{E}$, the values $f(\pi\cdot(r,\sigma),e)$ and $f(\pi'\cdot(r,\sigma),e)$ are both in $\{\top,\bot\}$ and are different.
Then $\mathtt{L}_i$ adds the suffix $e'=(r,\sigma)\cdot e$ to $\bm{E}$ and sets $f(\pi'',e')=\unk$ for all $\pi''\in\bm{S}\cup\bm{R}$.

\smallskip
\noindent
After either operation, $\mathtt{L}_i$ calls $\mathsf{Fill}(\mathbf{T})$.
Enforcing closedness may break consistency, and adding a distinguishing suffix for consistency may break closedness.
Thus, the learner repeats these operations until the table is prepared.

\subsection{Hypothesis Construction}
To construct a hypothesis from a prepared TOTU, $\mathtt{L}_i$ first uses an SMT solver to seek a concrete value in $\{\top,\bot\}$ for each entry whose value is $\unk$.
$\mathtt{L}_i$ defines a symbolic classification function $\mathcal{F}$ over table entries.
For each entry $(\pi,e)$, $\mathcal{F}(\pi,e)$ is a fresh symbolic Boolean if $f(\pi,e)=\unk$, and $\mathcal{F}(\pi,e)=f(\pi,e)$ otherwise.
Using $\mathcal{F}$, row equality is encoded as
\[
  \Psi_{\mathsf{row}}(\pi, \pi') \triangleq \bigwedge_{e \in \bm{E}} \bigl(\mathcal{F}(\pi, e) = \mathcal{F}(\pi', e)\bigr).
\]

\noindent
On this basis, the closedness and consistency of the table are encoded by
\vspace{-0.5mm}
\begin{align*}
  \Psi_{\mathsf{cls}} &\triangleq \bigwedge_{\pi \in \bm{R}}\;
    \bigvee_{\pi' \in \bm{S}} \Psi_{\mathsf{row}}(\pi,\, \pi'), \\[1pt]
  \Psi_{\mathsf{cst}} &\triangleq
    \bigwedge_{
    \pi\cdot(r,\sigma),\,\pi'\cdot(r,\sigma)\in\bm{S}\cup\bm{R}}
    \bigl(\Psi_{\mathsf{row}}(\pi,\pi')
    \Rightarrow
    \Psi_{\mathsf{row}}(\pi\cdot(r,\sigma),\,\pi'\cdot(r,\sigma))\bigr).
\end{align*}
\vspace{-2mm}

\noindent
The prefix-closure constraint of $\mathcal{L}(\mathcal{A}_i)$ is encoded by
\[
  \Psi_{\mathsf{pre}} \triangleq
    \bigwedge_{\substack{
      (\pi,e),(\pi',e')\in(\bm{S}\cup\bm{R})\times\bm{E}\\
      \pi\cdot e\text{ is a prefix of }\pi'\cdot e'}}
    \bigl(\mathcal{F}(\pi',e')
    \Rightarrow \mathcal{F}(\pi,e) \bigr).
\]

\noindent
In addition, entries inducing the same RTA region word are constrained to have the same value, as encoded by
\vspace{-0.5mm}
\[
  \Psi_{\mathsf{word}} \triangleq
  \bigwedge_{\substack{
  (\pi,e),(\pi',e')\in(\bm{S}\cup\bm{R})\times\bm{E}\\
  \pi\cdot e\;=\;\pi'\cdot e'}}
  \bigl(\mathcal{F}(\pi,e)=\mathcal{F}(\pi',e')\bigr).
\]

\noindent
A TOTU $\mathbf{T}$ is \emph{satisfiable} if an SMT solver finds the conjunction $\Psi_{\mathsf{cls}} \wedge \Psi_{\mathsf{cst}} \wedge \Psi_{\mathsf{pre}}\wedge \Psi_{\mathsf{word}}$ satisfiable.
In this case, a satisfying model induces a new classification function $f' : (\bm{S}\cup\bm{R})\times\bm{E}\to\{\top,\bot\}$.
For each RTA region word $\pi\in\bm{S}\cup\bm{R}$, let $\mathit{row}'(\pi)=(f'(\pi,e))_{e\in\bm{E}}$.
The hypothesis $\mathcal{H}_i$ is then constructed as follows.

\smallskip
\noindent\textbf{Stage 1 (location generation).}
For each $\pi\in\bm{S}$, $\mathtt{L}_i$ introduces a location $q_{\pi}$, with $q_{\epsilon}$ as the initial location.

\smallskip
\noindent\textbf{Stage 2 (transition generation).}
For each $\pi\in\bm{S}\cup\bm{R}$ and RTA region action $(r,\sigma)$ with $\pi\cdot(r,\sigma)\in\bm{S}\cup\bm{R}$, $\mathtt{L}_i$ finds $\pi',\pi''\in\bm{S}$ such that $\mathit{row}'(\pi')=\mathit{row}'(\pi)$ and $\mathit{row}'(\pi'')=\mathit{row}'(\pi\cdot(r,\sigma))$.
The learner then introduces a transition $(q_{\pi'},\sigma,\mathsf{guard}(x_i,d),q_{\pi''})$, where $d\in r$ is chosen arbitrarily.

\smallskip
\noindent\textbf{Stage 3 (guard arrangement).}
Given a transition guard $\phi=\mathsf{guard}(x_i,d)$, let $\phi^\ell$ be its lower-bound constraint, i.e., a constraint of the form $\kappa\le x_i$ or $\kappa<x_i$.
For each location $q_{\pi}$ and action $\sigma\in\Sigma_i$, the learner selects all transitions whose source is $q_{\pi}$ and action is $\sigma$, and then orders them as $\delta_1,\ldots,\delta_m$, where $\delta_j=(q_{\pi},\sigma,\phi_j,q_j)$ and $\phi_{j+1}^\ell\Rightarrow\phi_j^\ell$ for all $1\le j<m$.
Learner $\mathtt{L}_i$ replaces the guard of each $\delta_j$ by $\phi_j^\ell\wedge\neg\phi_{j+1}^\ell$\footnote{Here, $\neg(\kappa\le x_i)$ and $\neg(\kappa<x_i)$ are rewritten as $x_i<\kappa$ and $x_i\le\kappa$, respectively.} for $j<m$ and the guard of $\delta_m$ by $\phi_m^\ell$.

\smallskip
\noindent\textbf{Stage 4 (non-executable-location pruning).}
The value $f'(\pi,\epsilon)=\bot$ means that the learner currently treats every timed word $w$ with $\rabs{w}=\pi$ as outside $\mathcal{L}(\mathcal{A}_i)$.
Accordingly, the learner removes every location $q_{\pi}$ with $f'(\pi,\epsilon)=\bot$, together with all transitions whose source or target location is $q_{\pi}$, so that such words and their extensions are excluded from $\mathcal{L}(\mathcal{H}_i)$.

\smallskip
Notably, a prepared TOTU may still be unsatisfiable because potential row equality is nontransitive.
For example, there may exist $\pi\in\bm{S}$ and $\pi',\pi''\in\bm{R}$ with $\mathit{row}(\pi)\approx\mathit{row}(\pi')$ and $\mathit{row}(\pi)\approx\mathit{row}(\pi'')$, yet no classification function $f'$ satisfies both $\Psi_{\mathsf{row}}(\pi,\pi')$ and $\Psi_{\mathsf{row}}(\pi,\pi'')$.
In this case, $\mathtt{L}_i$ applies a fallback strategy that sets $\mathcal{H}_i=\mathcal{P}_i$, keeping the learning process moving.
As learning proceeds, newly collected information may make the table satisfiable again.

\subsection{Counterexample-Guided Refinement}
\label{Sec:LearnerCounterexample}

After issuing a component-level equivalence query, learner $\mathtt{L}_i$ receives $(w_{\mathit{ce}}^i,\beta)$ if the algorithm has not terminated.
Here $w_{\mathit{ce}}^i$ is a component-level counterexample and $\beta\in\{\top,\bot\}$.
Let $\pi_{\mathit{ce}}^i=\rabs{w_{\mathit{ce}}^i}$.
The learner refines $\mathbf{T}$ according to $\beta$.
To handle possible spurious counterexamples (cf. Section~\ref{Sec:CEClassification}), each learner $\mathtt{L}_i$ maintains a backtracking stack $\mathsf{BT}_i$, which is initially empty.
\begin{itemize}[topsep=3pt,itemsep=2pt,parsep=1pt]
  \item If $\beta=\bot$, the learner first pushes the pair $(\pi_{\mathit{ce}}^i,\mathbf{T}^{\mathit{snap}})$ onto $\mathsf{BT}_i$, where $\mathbf{T}^{\mathit{snap}}$ is a snapshot of $\mathbf{T}$.
		For each prefix $\pi'$ of $\pi_{\mathit{ce}}^i$ not in $\bm{S}\cup\bm{R}$, the learner adds $\pi'$ to $\bm{R}$ and sets $f(\pi',e)=\unk$ for all $e\in\bm{E}$.
		It then sets $f(\pi_{\mathit{ce}}^i,\epsilon)=\bot$.

  \item If $\beta=\top$, the learner first checks whether $\mathsf{BT}_i$ contains some pair $(\pi,\mathbf{T}^{\mathit{snap}})$ such that $\pi$ is a prefix of $\pi_{\mathit{ce}}^i$.
	  If so, $\mathbf{T}$ was previously updated based on a spurious counterexample.
	  The learner therefore chooses such a pair with the shortest $\pi$, rolls $\mathbf{T}$ back to the state recorded by $\mathbf{T}^{\mathit{snap}}$, and pops the selected pair and all pairs above it from $\mathsf{BT}_i$.
		  After this check, for each prefix $\pi'$ of $\pi_{\mathit{ce}}^i$ not in $\bm{S}\cup\bm{R}$, the learner adds $\pi'$ to $\bm{R}$ and sets $f(\pi',e)=\unk$ for all $e\in\bm{E}$.
		  It then sets $f(\pi_{\mathit{ce}}^i,\epsilon)=\top$.
\end{itemize}

%% file: Sections/6-Algorithm.tex
\section{Compositional Learning Algorithm}
\label{Sec:Algorithm}

\input{Algorithms/MainLoop}

Based on the preceding components, Algorithm~\ref{alg:Main} gives the overall compositional learning algorithm.
It interacts with the system-level teacher for the target network $\mathcal{N}=(\mathcal{A}_1,\ldots,\mathcal{A}_n)$ and takes the number $n$ of components and the component alphabets $\Sigma_1,\ldots,\Sigma_n$ as input.
When it terminates, the timed language of the returned network $\mathcal{N}'=(\mathcal{H}_1,\ldots,\mathcal{H}_n)$ is equivalent to that of $\mathcal{N}$.
At initialization, each learner $\mathtt{L}_i$ creates its initial TOTU $\mathbf{T}_i$, and the adapter initializes the positive-evidence RTA $\mathcal{P}_i$ for component $\mathcal{A}_i$ (Lines~\ref{lines:main_init}--\ref{lines:main_pe_init}).
The algorithm then proceeds in rounds.
In each round, every learner first calls $\mathsf{Fill}(\mathbf{T}_i)$, makes $\mathbf{T}_i$ prepared, and constructs $\mathcal{H}_i$ from $\mathbf{T}_i$ if it is satisfiable; otherwise, it sets $\mathcal{H}_i=\mathcal{P}_i$ (Lines~\ref{lines:main_fill}--\ref{lines:main_unsat}).
After all component hypotheses are constructed, the adapter forms $\mathcal{N}'=(\mathcal{H}_1,\ldots,\mathcal{H}_n)$ and submits a system-level equivalence query to the teacher (Lines~\ref{lines:main_network}--\ref{lines:main_eq}).
If the teacher returns $\top$, the algorithm terminates and returns $\mathcal{N}'$.
Otherwise, the teacher returns a system-level counterexample $u_{\mathit{ce}}$.
For each component, the adapter derives a component-level answer $(w_{\mathit{ce}}^i,\beta)$ from $u_{\mathit{ce}}$, and learner $\mathtt{L}_i$ refines $\mathbf{T}_i$ using $(w_{\mathit{ce}}^i,\beta)$ (Lines~\ref{lines:main_component_eq}--\ref{lines:main_table_refine}).
After all learners have refined their tables, the next round starts.
Algorithm~\ref{alg:Main} returns only after the system-level teacher answers $\top$ to the equivalence query on $\mathcal{N}'$ (Lines~\ref{lines:main_loop_end}--\ref{lines:main_return}), yielding the following correctness result.
\begin{theorem}
\label{theorem:PartialCorrectness}
If Algorithm~\ref{alg:Main} terminates and returns a hypothesis network $\mathcal{N}'$, then $\mathcal{L}(\mathcal{N}')=\mathcal{L}(\mathcal{N})$.
\end{theorem}

Unlike learning algorithms for untimed synchronous systems~\cite{Compositional_DFA_learning} and individual RTAs~\cite{RTA_learning}, Algorithm~\ref{alg:Main} does not guarantee termination.
For an individual RTA, termination relies on exact query answers and the finite behavioral partition induced by its \emph{clock ceiling}, i.e., the largest integer constant in its guards.
In our compositional setting, however, unknown answers and possible spurious counterexamples may cause repeated refinement and backtracking, so the same convergence argument does not apply.
Nevertheless, the algorithm successfully completes most benchmark instances within the resource limits (cf. Section~\ref{Sec:Experiment}).
As future work, we will investigate general termination guarantees, which may require additional prior knowledge about the target systems.

The learned components need not be individually language-equivalent to their targets, and such equivalence is generally unidentifiable from system-level queries.
For example, Fig.~\ref{fig:non-identifiability} compares the target network $\mathcal{N}=(\mathcal{A}_1,\mathcal{A}_2)$ with a hypothesis network $\mathcal{N}'=(\mathcal{H}_1,\mathcal{H}_2)$, where $\mathcal{H}_1=\mathcal{A}_1$.

\input{Figures/NonIdentifiability}

\noindent
By the definition of the timed language of an NRTA, in both $\mathcal{N}$ and $\mathcal{N}'$, an action $a$ can occur exactly when its preceding delay exceeds five time units.
Hence, $\mathcal{L}(\mathcal{N}')=\mathcal{L}(\mathcal{N})$.
However, $\mathcal{L}(\mathcal{H}_2)\ne\mathcal{L}(\mathcal{A}_2)$, since a timed word such as $w=(3.5,a)(3.5,a)$ belongs only to $\mathcal{L}(\mathcal{A}_2)$.
The system-level teacher therefore cannot return a counterexample that exposes this component-level difference.

Despite this limitation, for any subset of components, the target and learned subnetworks have the same timed behaviors over actions not used by components outside the subset.
This allows verification of such behaviors to focus on the selected components.
The following lemma formalizes this property.

\begin{lemma}
\label{lemma:subnetwork-preservation}
Suppose Algorithm~\ref{alg:Main} terminates with $\mathcal{N}'$.
For every nonempty $I\subseteq\{1,\ldots,n\}$, let $\mathcal{N}_I=(\mathcal{A}_i)_{i\in I}$, $\mathcal{N}'_I=(\mathcal{H}_i)_{i\in I}$, and $\Sigma_I^{\mathit{loc}}=\Sigma\setminus\bigcup_{j\notin I}\Sigma_j$.
Then,
\vspace{-1mm}
\[
  \mathcal{L}(\mathcal{N}_I)\cap(\NR\times\Sigma_I^{\mathit{loc}})^*
  =
  \mathcal{L}(\mathcal{N}'_I)\cap(\NR\times\Sigma_I^{\mathit{loc}})^*.
\]
\end{lemma}
The proof is provided in Appendix~\ref{app:proof-subnetwork-preservation}.

%% file: Algorithms/MainLoop.tex
\begin{compactalgorithm}[t]{-3mm}
  \caption{Learning NRTA from a system-level teacher}
  \label{alg:Main}
  \SetKwInOut{Input}{input}
  \SetKwInOut{Output}{output}

  \Input{the number $n$ of components; component alphabets $\Sigma_1,\ldots,\Sigma_n$.}
	\Output{upon termination, an NRTA $\mathcal{N}'=(\mathcal{H}_1,\ldots,\mathcal{H}_n)$ with $\mathcal{L}(\mathcal{N}')=\mathcal{L}(\mathcal{N})$.}

  initialize learners $\mathtt{L}_1,\ldots,\mathtt{L}_n$ with TOTUs $\mathbf{T}_1,\ldots,\mathbf{T}_n$\;
  \label{lines:main_init}
	  initialize the adapter with RTAs $\mathcal{P}_1,\ldots,\mathcal{P}_n$\;
	  \label{lines:main_pe_init}
			  \Repeat{$\mathit{ans}=\top$}{
			    \label{lines:main_loop}
				    \For{$i \gets 1$ \KwTo $n$}{
		      \label{lines:main_for}
			      $\mathsf{Fill}(\mathbf{T}_i)$\;
			      \label{lines:main_fill}
			      make $\mathbf{T}_i$ prepared\;
			      \label{lines:main_prepare}
				\textbf{if} $\mathbf{T}_i$ is satisfiable \textbf{then} construct $\mathcal{H}_i$ from $\mathbf{T}_i$\;
  				\textbf{else} $\mathcal{H}_i\leftarrow\mathcal{P}_i$\;
		      \label{lines:main_sat}
		      \label{lines:main_unsat}
				    }
			    \label{lines:main_barrier}
			    $\mathcal{N}' \gets (\mathcal{H}_1,\ldots,\mathcal{H}_n)$\;
	    \label{lines:main_network}
			$\mathit{ans}\gets\mathsf{EQ}(\mathcal{N}')$\;
	      \label{lines:main_eq}
	    \If{$\mathit{ans}=({\bot},u_{\mathit{ce}})$}{
	      \label{lines:main_not_eq}
	      \For{$i \gets 1$ \KwTo $n$}{
					        derive $(w_{\mathit{ce}}^i,\beta)$ from $u_{\mathit{ce}}$\;
			        \label{lines:main_component_eq}
			        refine $\mathbf{T}_i$ using $(w_{\mathit{ce}}^i,\beta)$\;
		        \label{lines:main_negative_refine}
		        \label{lines:main_positive_refine}
		      }
	      \label{lines:main_table_refine}
		    }
		  }
		  \label{lines:main_loop_end}
		  \Return{$\mathcal{N}'$}\;
	  \label{lines:main_return}
\end{compactalgorithm}

%% file: Figures/NonIdentifiability.tex
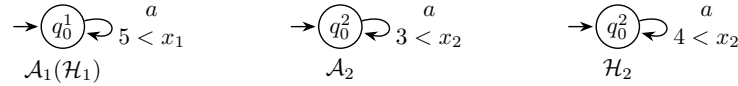
\begin{figure}[H]
\vspace{-5mm}
\centering
\scalebox{0.88}{%
\begin{tikzpicture}[
  ->,
  >=Stealth,
  shorten >=1pt,
  auto,
  node distance=3.5cm,
  semithick,
  every node/.style={font=\normalsize},
  every state/.style={minimum size=2.0em,inner sep=1pt},
  initial text={}
]
  \node[initial,state] (a1) {$q_0^1$};
  \node[initial,state,right=of a1] (h1) {$q_0^2$};
  \node[initial,state,right=of h1] (a2) {$q_0^2$};

  \path
    (a1) edge[loop right] node[right,align=center] {$a$\\$5<x_1$} (a1)
    (h1) edge[loop right] node[right,align=center] {$a$\\$3<x_2$} (h1)
    (a2) edge[loop right] node[right,align=center] {$a$\\$4<x_2$} (a2);

  \node[below=1pt of a1] {$\mathcal{A}_1 (\mathcal{H}_1)$};
  \node[below=1pt of h1] {$\mathcal{A}_2$};
  \node[below=1pt of a2] {$\mathcal{H}_2$};
\end{tikzpicture}
}
\vspace{-4.5mm}
\caption{Target and hypothesis component RTAs.}
\label{fig:non-identifiability}
\vspace{-6mm}
\end{figure}

%% file: Sections/7-Experiment.tex
\section{Evaluation}
\label{Sec:Experiment}

We implement the compositional learning approach in Java with Z3 (v4.14.1.0) for SMT solving.
We denote our approach by $\mathsf{CL}$ and compare it with two monolithic approaches that learn a single automaton, namely one based on \emph{timed observation tables} ($\mathsf{TOT}$)~\cite{HSCC24_Multi_Clock} and one based on \emph{timed classification trees} ($\mathsf{TCT}$)~\cite{Science_China_TTT}.
As public implementations of the baselines are unavailable, we reimplement them from their published algorithms in Java and provide the prior information required by their respective learning settings.
All three approaches use the same system-level teacher backend.
Experiments run on an Ubuntu~22.04 virtual machine with an AMD EPYC~9754 processor.
The artifact repository provides the source code, benchmark instances, and raw data~\cite{ArtifactRepository}. Further experimental setup details and an illustrative learning example are provided in Appendices~\ref{app:teacher-baselines} and~\ref{app:illustrative-example}, respectively.

\subsection{Randomly Generated Benchmark Instances}
\label{sec:Random_Exp}

We first evaluate the approaches on eight groups of randomly generated synchronous RTA networks, each containing 30 benchmark instances.
Each run is subject to a 300-second timeout and a 16-GB memory limit.
\input{Tables/RandomCaseResults}

Table~\ref{tab:random-groups} reports the group-wise and overall results.
For each group, $n$ and $|\Sigma|$ give the fixed number of components and the fixed size of the system alphabet, while $|Q_i|$ is the fixed number of locations in each component.
The quantities $\overline{|\Delta_i|}$ and $\overline{K_i}$ are averaged over all components and benchmark instances in the group, where $|\Delta_i|$ is the number of transitions and $K_i$ is the largest integer constant in the guards of component $\mathcal{A}_i$.
``Solved'' is the number of benchmark instances successfully learned within the time and memory limits.
All other runs timed out without exhausting the available memory.
PAR-1 is the average running time over all benchmark instances, with each timeout assigned the 300-second cutoff.
Using G1 as the baseline, G2--G6 increase $n$, $|\Sigma|$, $|Q_i|$, $\overline{|\Delta_i|}$, and $\overline{K_i}$, respectively.
G7 and G8 simultaneously increase several structural parameters to form larger-scale benchmark instances.

Across all groups, $\mathsf{TCT}$ consistently outperforms $\mathsf{TOT}$, which is consistent with the reported efficiency advantage~\cite{Science_China_TTT}.
$\mathsf{CL}$ achieves success rates comparable to $\mathsf{TCT}$ on G1--G6, while solving substantially more instances than both monolithic approaches on the larger-scale groups G7 and G8.
One possible explanation is that monolithic learning becomes more costly as the target system grows, particularly with increases in the number of components and component transitions, as suggested by G2 and G5.
By learning component hypotheses separately, $\mathsf{CL}$ may mitigate the effect of growth along these dimensions.
This potential advantage comes with additional processing for unknown entries and possible spurious counterexamples, which may offset the benefits of compositional learning on some instances.
For instance, backtracking occurs in 112 of the 213 successful $\mathsf{CL}$ runs.
Across these 213 runs, the mean and maximum backtracking counts are 1.30 and 11, respectively.

\input{Figures/RandomSolvedOverlap}
Figure~\ref{fig:random-solved-overlap} shows the overlap of the instances solved by the three approaches across all benchmark instances (areas not to scale).
Overall, $\mathsf{CL}$ solves 213 of the 240 (88.8\%) instances, compared with 89 for $\mathsf{TOT}$ and 175 for $\mathsf{TCT}$.
$\mathsf{CL}$ does not solve 13 instances solved by at least one monolithic approach, while uniquely solving 51 others.
On the 85 jointly solved instances, $\mathsf{CL}$ averages 14.3 system-level equivalence queries, fewer than $\mathsf{TOT}$ (33.8) and $\mathsf{TCT}$ (28.2), while its 367.7 system-level membership queries lie between $\mathsf{TOT}$ (214.6) and \mbox{$\mathsf{TCT}$ (435.2)}.
Since answering a system-level equivalence query generally requires more computation than answering a system-level membership query, the lower system-level equivalence query count may partly explain $\mathsf{CL}$'s runtime advantage on some instances.
Additional analyses are provided in Appendix~\ref{app:evaluation}.

\input{Tables/RealCaseResults}

\subsection{Adapted Case-Study Benchmarks}
\label{sec:Real_Exp}

We further evaluate the approaches on six benchmark instances adapted from case studies reported in the literature~\cite{non_RTA,AGR_2023} and examples provided by verification tools~\cite{Andre2016PartialOrder,uppaaltutorial}.
The adaptations preserve component structure and relevant timed synchronization while re-encoding non-NRTA features, such as \emph{location invariants} and \emph{shared variables}, as documented in the artifact repository~\cite{ArtifactRepository}.
Each run is subject to a 600-second timeout and a 16-GB memory limit.

Table~\ref{tab:real-cases} reports model sizes and learning times for each instance.
The notation follows Table~\ref{tab:random-groups}, except that $\overline{|Q_i|}$ reports the average number of locations per component because component sizes may vary within an instance.
All running times are reported in seconds, and ``TO'' denotes a timeout.
$\mathsf{CL}$ solves all six instances, while $\mathsf{TCT}$ solves two and $\mathsf{TOT}$ solves one.
On Fischer, all three approaches terminate, with both monolithic approaches requiring substantially less learning time than $\mathsf{CL}$.
Although these results do not demonstrate a uniform runtime advantage, they support the feasibility of compositional learning across the diverse synchronization structures in these case studies.

%% file: Tables/RandomCaseResults.tex
\begin{compacttable}[t]{-3mm}
\centering
\caption{Experimental results on learning randomly generated benchmark instances}
\vspace{-1mm}
\label{tab:random-groups}
\setlength{\tabcolsep}{6pt}
\resizebox{0.98\textwidth}{!}{%
\begin{tabular}{cccccc ccc rrr}
\toprule
\multirow{2}{*}{Group} & \multirow{2}{*}{$n$} & \multirow{2}{*}{$|\Sigma|$} & \multirow{2}{*}{$|Q_i|$} & \multirow{2}{*}{$\overline{|\Delta_i|}$} & \multirow{2}{*}{$\overline{K_i}$} & \multicolumn{3}{c}{Solved} & \multicolumn{3}{c}{PAR-1 (in seconds)} \\
\cmidrule(lr){7-9}\cmidrule(lr){10-12}
& & & & & & $\mathsf{TOT}$ & $\mathsf{TCT}$ & $\mathsf{CL}$ & \multicolumn{1}{c}{$\mathsf{TOT}$} & \multicolumn{1}{c}{$\mathsf{TCT}$} & \multicolumn{1}{c}{$\mathsf{CL}$} \\
\midrule
G1 & 2 & 3 & 2 & 3.88 & 7.85 & 21 & 30 & 28 & 98.587 & 6.107 & 31.521 \\
G2 & 3 & 3 & 2 & 3.77 & 7.78 & 7 & 20 & 26 & 240.342 & 123.023 & 72.827 \\
G3 & 2 & 4 & 2 & 3.73 & 7.83 & 17 & 30 & 30 & 149.354 & 12.268 & 13.989 \\
G4 & 2 & 3 & 3 & 4.08 & 7.67 & 16 & 30 & 28 & 146.130 & 6.606 & 31.840 \\
G5 & 2 & 3 & 2 & 7.93 & 7.97 & 9 & 23 & 26 & 213.846 & 104.549 & 60.191 \\
G6 & 2 & 3 & 2 & 3.85 & 15.85 & 16 & 30 & 29 & 143.062 & 21.598 & 19.637 \\
G7 & 3 & 4 & 3 & 5.96 & 7.93 & 0 & 8 & 23 & 300.000 & 236.655 & 108.065 \\
G8 & 4 & 5 & 2 & 3.87 & 7.89 & 3 & 4 & 23 & 276.780 & 260.566 & 97.315 \\
\midrule
Overall & -- & -- & -- & -- & -- & 89 & 175 & 213 & 196.013 & 96.422 & 54.423 \\
\bottomrule
\end{tabular}
}
\end{compacttable}

%% file: Figures/RandomSolvedOverlap.tex
\begin{wrapfigure}{r}{0.36\textwidth}
\vspace{-6mm}
\centering
\begin{tikzpicture}[x=0.62cm,y=0.62cm,every node/.style={font=\small}]
  \draw[thick,draw=blue!70!black,fill=blue!25,fill opacity=0.28]
    (0,0) ellipse [x radius=2.1,y radius=1.3];
  \draw[thick,draw=red!70!black,fill=red!25,fill opacity=0.28]
    (2.3,0) ellipse [x radius=2.1,y radius=1.3];
  \draw[thick,draw=green!50!black,fill=green!30,fill opacity=0.55]
    (1.8,-0.35) ellipse [x radius=1.25,y radius=0.75];

  \node[text=blue!70!black] at (-0.2,0.9) {$\mathsf{CL}$};
  \node[text=red!70!black] at (2.5,0.9) {$\mathsf{TCT}$};
  \node[text=green!40!black] at (1.8,0) {$\mathsf{TOT}$};

  \node at (-1.0,0) {51};
  \node at (1.1,0.65) {77};
  \node at (1.3,-0.55) {85};
  \node at (2.4,-0.55) {4};
  \node at (3.75,0) {9};
  \node at (3.25,-1.7) {Unsolved: 14};
\end{tikzpicture}
\vspace{-9mm}
\caption{Solved-instance overlap}
\label{fig:random-solved-overlap}
\vspace{-5.5mm}
\end{wrapfigure}
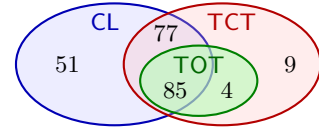

%% file: Tables/RealCaseResults.tex
\begin{compacttable}[t]{-5mm}
\centering
\caption{Experimental results on learning adapted case-study benchmark instances}
\vspace{-1mm}
\label{tab:real-cases}
\setlength{\tabcolsep}{9pt}
\resizebox{0.98\textwidth}{!}{%
\begin{tabular}{cccccc rrr}
\toprule
Benchmark & $n$ & $|\Sigma|$ & $\overline{|Q_i|}$ & $\overline{|\Delta_i|}$ & $\overline{K_i}$ & \multicolumn{1}{c}{$\mathsf{TOT}$} & \multicolumn{1}{c}{$\mathsf{TCT}$} & \multicolumn{1}{c}{$\mathsf{CL}$} \\
\midrule
AUTOSAR & 9 & 7 & 3.11 & 4.78 & 3.33 & \multicolumn{1}{c}{TO} & \multicolumn{1}{c}{TO} & 39.956 \\
Fischer & 3 & 10 & 3.67 & 8.00 & 1.33 & 65.538 & 4.488 & 469.573 \\
IC Final-Test & 2 & 9 & 9.00 & 16.00 & 12.50 & \multicolumn{1}{c}{TO} & \multicolumn{1}{c}{TO} & 69.405 \\
2Doors & 4 & 6 & 5.00 & 6.50 & 4.00 & \multicolumn{1}{c}{TO} & \multicolumn{1}{c}{TO} & 215.970 \\
Bridge & 4 & 6 & 4.50 & 9.00 & 1.75 & \multicolumn{1}{c}{TO} & \multicolumn{1}{c}{TO} & 70.606 \\
Not-Not-And & 5 & 8 & 4.00 & 6.00 & 1.40 & \multicolumn{1}{c}{TO} & 294.141 & 9.107 \\
\bottomrule
\end{tabular}
}
\end{compacttable}

%% file: Sections/8-Related_Work.tex
\section{Related Work}
\label{Sec:Related}
\textbf{Active learning of timed systems} has been studied for target models such as real-time automata~\cite{RTA_learning,non_RTA}, Mealy machines with timers~\cite{mealy_multi_timers,vaandrager2023learning}, and timed automata~\cite{DOTA_learning,HSCC24_Multi_Clock,CAV2023_Multi_Clock}.
Tree-based data structures have also been introduced to accelerate timed learning~\cite{TTT_TA,Science_China_TTT,timed_L_sharp}.
These works usually rely on an idealized teacher but construct a monolithic automaton, whereas our approach constructs a network of component RTAs using only system-level queries.
The component learners are inspired by the learner design of RTA learning~\cite{RTA_learning}.
However, the compositional learning target requires additional mechanisms to deal with unknown answers and possible spurious counterexamples, such as TOTUs with unknown entries, SMT-based hypothesis construction, and adapter-driven refinement.

\smallskip
\noindent
\textbf{Compositional active learning} has mainly focused on untimed systems.
Some works consider settings with decomposable system outputs~\cite{product_automata}, no synchronization among components~\cite{Compsitional_Interleaving}, or direct access to individual black-box components~\cite{FujinamiWASIIH25}.
For untimed synchronous systems with the number of components and their alphabets known, prior work introduced a compositional learning architecture comprising component learners, an adapter, and a system-level teacher~\cite{Compositional_DFA_learning}.
Our approach adopts this architecture for timed synchronous systems and addresses the resulting challenges by designing a query adapter equipped with positive-evidence RTAs and TOTU-based component learners.
Recent untimed work further reduces prior knowledge of the number of components and their alphabets~\cite{henry2025compositional,samadi2025compositional}, which is a future direction for timed compositional learning.

%% file: Sections/9-Conclusion.tex
\section{Conclusion}
\label{Sec:Conclusion}

In this paper, we present a compositional active learning approach for timed synchronous systems using only system-level access.
It combines a query adapter equipped with positive-evidence RTAs and TOTU-based component learners.
Upon termination, the approach learns a deterministic RTA for each component, and these RTAs form a network whose timed language is equivalent to that of the target system.
Within the respective resource limits, our approach terminates successfully on 213 of the 240 (88.8\%) randomly generated benchmark instances and all six adapted case-study benchmark instances.
It also requires less learning time than the monolithic approaches on some instances.
Future work includes characterizing termination conditions for the learning algorithm, reducing prior knowledge about the number of components and their alphabets, and extending the approach to more general timed component models.

%% file: Sections/Appendix.tex
\input{Sections/Appendix/Projection_Lemma}
\input{Sections/Appendix/RTA_Region_Equality}
\input{Sections/Appendix/Subnetwork_Language_Preservation}
\input{Sections/Appendix/Teacher_and_Baseline_Settings}
\input{Sections/Appendix/Additional_Evaluation}
\input{Sections/Appendix/Illustrative_Example}

%% file: Sections/Appendix/Projection_Lemma.tex
\section{Proof of Lemma~\ref{lemma:proj_sound}}
\label{app:proof-projection}

Let $u=(d_1,\sigma_1)\cdots(d_m,\sigma_m)$.
For a component $\mathcal{A}_i$, let $j_1<\cdots<j_k$ be exactly the indices such that $\sigma_{j_\ell}\in\Sigma_i$, and let $j_0=0$.
By the definition of projection,
\[
  \project{u}{i}
  = (D_1,\sigma_{j_1})\cdots(D_k,\sigma_{j_k}),
  \qquad
  D_\ell=\sum_{h=j_{\ell-1}+1}^{j_\ell}d_h.
\]
If there is no such index, then $\project{u}{i}=\epsilon$, which belongs to $\mathcal{L}(\mathcal{A}_i)$.

\medskip
\noindent\textbf{($\Rightarrow$)}
Assume $u\in\mathcal{L}(\mathcal{N})$ and fix a run of $\mathcal{N}$ on $u$.
Between the actions at indices $j_{\ell-1}$ and $j_\ell$, component $\mathcal{A}_i$ takes no transition.
Its location remains unchanged, while the value of $x_i$ continues to increase as time elapses.
Since $x_i$ is initially zero and reset at $j_{\ell-1}$, where $j_0=0$ denotes the initial instant, the value of $x_i$ when $\sigma_{j_\ell}$ occurs, before reset, is
\[
  (v_{j_\ell-1}+d_{j_\ell})(x_i)
  = \sum_{h=j_{\ell-1}+1}^{j_\ell}d_h
  = D_\ell.
\]
By the NRTA run condition, the global step at $j_\ell$ contains a $\sigma_{j_\ell}$-transition of $\mathcal{A}_i$ whose guard is satisfied when $x_i$ has value $D_\ell$.
The transitions extracted at $j_1,\ldots,j_k$ therefore form a run of $\mathcal{A}_i$ on $\project{u}{i}$.
Applying the same argument to every component gives $\project{u}{i}\in\mathcal{L}(\mathcal{A}_i)$ for all $1\le i\le n$.

\medskip
\noindent\textbf{($\Leftarrow$)}
Assume $\project{u}{i}\in\mathcal{L}(\mathcal{A}_i)$ for every component and, for each $i$, fix a run $\rho_i$ of $\mathcal{A}_i$ on $\project{u}{i}$.
We construct a run $\rho$ of $\mathcal{N}$ on $u$ by induction.

After a global prefix has been performed in $\rho$, the location of each component $\mathcal{A}_i$ is the location reached in $\rho_i$ after the actions of that prefix in $\Sigma_i$.
Moreover, the value of $x_i$ equals the sum of the delays since the most recent such action, or since the initial instant if none has occurred.
This correspondence holds for the empty prefix.

Suppose the correspondence holds after the first $j-1$ actions and consider $(d_j,\sigma_j)$.
For every $i$ with $\sigma_j\in\Sigma_i$, the value of $x_i$ immediately before $\sigma_j$ is exactly the next accumulated delay in $\project{u}{i}$.
The next transition in $\rho_i$ is thus an enabled $\sigma_j$-transition from its current location.
All participating components can take these transitions simultaneously, while every other component keeps its location.

The update $v_j=(v_{j-1}+d_j)[\sigma_j]$ resets exactly the clocks of the participating components.
It therefore preserves the correspondence after the $j$-th action.
By induction, this construction yields a run $\rho$ of $\mathcal{N}$ on $u$, and hence $u\in\mathcal{L}(\mathcal{N})$.

%% file: Sections/Appendix/RTA_Region_Equality.tex
\section{Proof of Lemma~\ref{lemma:rta_region_equal}}
\label{app:proof-region-equality}

We say that two delays $d,d'\in\NR$ lie in the same \emph{RTA delay region} if, for some $k\in\mathbb{N}$, either $d=d'=k$ or $d,d'\in(k,k+1)$.
Then, for any atomic constraint
\[
  \theta\in\{x_i\le\kappa,\ x_i<\kappa,\ \kappa\le x_i,\ \kappa<x_i\mid \kappa\in \mathbb{N}\},
\]
it follows that
\[
  \{x_i\mapsto d\}\models\theta
  \quad\Longleftrightarrow\quad
  \{x_i\mapsto d'\}\models\theta.
\]
Since every guard $\phi\in\Phi(\{x_i\})$ is a conjunction of such atomic constraints, this equivalence extends to every guard in $\Phi(\{x_i\})$.
Let $w=(d_1,\sigma_1)\cdots(d_m,\sigma_m)$.
Let $w'=(d'_1,\sigma'_1)\cdots(d'_m,\sigma'_m)$.

\medskip
\noindent\textbf{($\Rightarrow$)}
Since \mbox{$\rabs{w}=\rabs{w'}$}, we have $\sigma_j = \sigma_j'$ and $d_j,d'_j$ lie in the same RTA delay region for every $1\le j\le m$.
Suppose $w\in\mathcal{L}(\mathcal{A}_i)$.
Then, $\mathcal{A}_i$ has a run on $w$ using transitions $(q_{j-1},\sigma_j,\phi_j,q_j)$ with $\{x_i\mapsto d_j\}\models\phi_j$ for every $1\le j\le m$.
By the guard equivalence established above, $\{x_i\mapsto d'_j\}\models\phi_j$ whenever $\{x_i\mapsto d_j\}\models\phi_j$.
Hence, the same transition sequence is also a run on $w'$, and $w'\in\mathcal{L}(\mathcal{A}_i)$.

\medskip
\noindent\textbf{($\Leftarrow$)}
The reverse implication follows by exchanging $w$ and $w'$ in the preceding argument, giving $w'\in\mathcal{L}(\mathcal{A}_i)\Rightarrow w\in\mathcal{L}(\mathcal{A}_i)$.

%% file: Sections/Appendix/Subnetwork_Language_Preservation.tex
\section{Proof of Lemma~\ref{lemma:subnetwork-preservation}}
\label{app:proof-subnetwork-preservation}

Fix a nonempty set $I\subseteq\{1,\ldots,n\}$ and a timed word $u\in(\NR\times\Sigma_I^{\mathit{loc}})^*$.
For every $j\notin I$, the definition of $\Sigma_I^{\mathit{loc}}$ gives $\project{u}{j}=\epsilon$.
The empty word belongs to the language of every RTA, so $\project{u}{j}\in\mathcal{L}(\mathcal{A}_j)$ and $\project{u}{j}\in\mathcal{L}(\mathcal{H}_j)$ for all $j\notin I$.
By Lemma~\ref{lemma:proj_sound} applied to $\mathcal{N}$ and $\mathcal{N}_I$,
\[
  u\in\mathcal{L}(\mathcal{N})
  \quad\Longleftrightarrow\quad
  u\in\mathcal{L}(\mathcal{N}_I).
\]
The same argument applied to $\mathcal{N}'$ and $\mathcal{N}'_I$ gives
\[
  u\in\mathcal{L}(\mathcal{N}')
  \quad\Longleftrightarrow\quad
  u\in\mathcal{L}(\mathcal{N}'_I).
\]

\noindent
By the premise, Algorithm~\ref{alg:Main} has terminated with $\mathcal{L}(\mathcal{N}')=\mathcal{L}(\mathcal{N})$, so the preceding equivalences yield
\[
  u\in\mathcal{L}(\mathcal{N}_I)
  \quad\Longleftrightarrow\quad
  u\in\mathcal{L}(\mathcal{N}'_I).
\]
The claim follows because $u$ was chosen arbitrarily from $(\NR\times\Sigma_I^{\mathit{loc}})^*$.

%% file: Sections/Appendix/Teacher_and_Baseline_Settings.tex
\section{Teacher Implementation and Baseline Settings}
\label{app:teacher-baselines}

\subsection{System-Level Teacher}

All three approaches use the same model-based system-level teacher, which has direct access to the target NRTA $\mathcal{N}$.
For a membership query, the teacher checks whether the queried timed word belongs to $\mathcal{L}(\mathcal{N})$.
For an equivalence query, it performs a breadth-first symbolic exploration of the target and hypothesis networks, representing zones by difference-bound matrices (DBMs)~\cite{Dill1989,LarsenPY97}.
The exploration checks their language difference in both directions and reconstructs a counterexample with the fewest timed actions when a difference is found.
If the exploration finds no difference, the teacher returns $\top$.
Using the same backend for all approaches avoids differences caused by the teacher implementation.

\subsection{Baseline Settings}

Both $\mathsf{TOT}$ and $\mathsf{TCT}$ learn the target network $\mathcal{N}=(\mathcal{A}_1,\ldots,\mathcal{A}_n)$ as a single timed automaton with $n$ clocks and require the following information:
\begin{enumerate}[topsep=5pt,itemsep=1pt,parsep=1pt]
  \item the system alphabet $\Sigma$;
  \item the number $n$ of clocks;
  \item the clock ceiling $K_i$ for each clock $x_i$.
\end{enumerate}
The first two requirements follow from $\Sigma=\bigcup_{i=1}^{n}\Sigma_i$ and the fact that every RTA component has one clock.
The clock ceiling $K_i$ is the largest integer constant in the guards of component $\mathcal{A}_i$.
The region constructions in both original algorithms depend on these ceilings.
We therefore supply the exact $K_i$ values to both baselines, whereas $\mathsf{CL}$ neither receives nor uses them.
In this paper, the timed language of an NRTA is prefix-closed, and $\mathsf{CL}$ exploits this property in its TOTU operations.
The original $\mathsf{TOT}$ and $\mathsf{TCT}$ algorithms instead target general timed languages defined by accepting locations and do not assume prefix closure.
We retain their original learning settings in our implementations.

Moreover, the \emph{smart-teacher} variant of $\mathsf{TOT}$ used in our experiments and $\mathsf{TCT}$ require the teacher to answer \emph{reset-information queries}, which identify the clocks reset by each action in a timed word.
For an NRTA, each component resets its clock whenever it participates in an action.
Hence, the reset set of an action $\sigma$ is $\{x_i\mid\sigma\in\Sigma_i\}$ and can be derived directly from the known component alphabets.
Therefore, reset-information queries require no additional teacher capability in our setting.

%% file: Sections/Appendix/Additional_Evaluation.tex
\section{Additional Analyses for Randomly Generated Benchmark Instances}
\label{app:evaluation}

\subsection{Sensitivity to Timeout Penalties}

For a timeout $T$ and a penalty factor $k$, PAR-$k$ is defined as
\[
  \mathrm{PAR}\text{-}k
  = \frac{\sum_{c\in C_{\mathit{solved}}} T_c
  + kT\,|C_{\mathit{timeout}}|}{|C|},
\]
where $C$ is the set of all benchmark instances, $C_{\mathit{solved}}$ and $C_{\mathit{timeout}}$ are its solved and timed-out subsets, respectively, and $T_c$ is the running time on a solved instance $c$.
Table~\ref{tab:random-par-sensitivity} reports the results for $T=300$ seconds and three penalty factors.
All PAR values are reported in seconds.
The ordering of the approaches is unchanged for all three penalty factors, indicating that the comparison is robust to these choices.

\vspace{-5mm}
\input{Tables/PARSensitivity}
\vspace{-10mm}

\subsection{Cumulative Solved Benchmark Instances}

Figure~\ref{fig:random-cactus} shows the cumulative number of randomly generated benchmark instances solved within each running-time budget.
The linear scale in Fig.~\ref{fig:random-cactus}(a) shows behavior over the full time budget, while the logarithmic scale in Fig.~\ref{fig:random-cactus}(b) expands the early completion times.
Both plots use observed completion times, with timed-out runs remaining unsolved.

\vspace{-3mm}
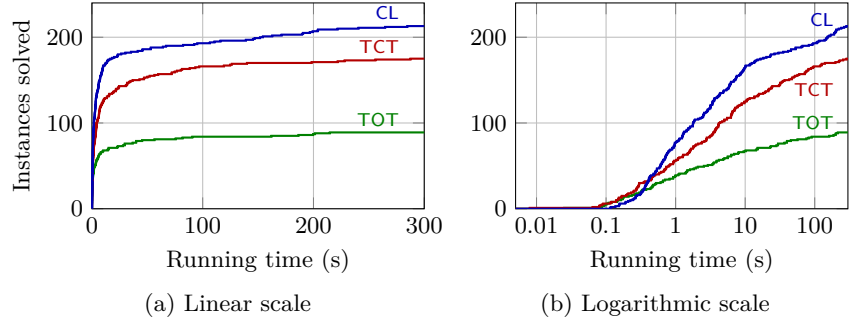
\begin{figure}[H]
\centering
\begin{minipage}[t]{0.49\textwidth}
\centering
\begin{tikzpicture}
\begin{axis}[
  width=\linewidth,
  height=0.72\linewidth,
  xmin=0,
  xmax=300,
  xtick={0,100,200,300},
  ymin=0,
  ymax=240,
  xlabel={Running time (s)},
  ylabel={Instances solved},
  grid=major
]
\addplot[thick, const plot, mark=none,green!50!black] table[x=time,y=solved] {Figures/random_cactus_TOT.dat};
\addplot[thick, const plot, mark=none,red!70!black] table[x=time,y=solved] {Figures/random_cactus_TCT.dat};
\addplot[thick, const plot, mark=none,blue!70!black] table[x=time,y=solved] {Figures/random_cactus_CL.dat};
\node[font=\scriptsize,anchor=south east,text=green!50!black,yshift=1pt] at (axis cs:285,83) {$\mathsf{TOT}$};
\node[font=\scriptsize,anchor=south east,text=red!70!black,yshift=1pt] at (axis cs:285,168) {$\mathsf{TCT}$};
\node[font=\scriptsize,anchor=south east,text=blue!70!black,yshift=1pt] at (axis cs:285,208) {$\mathsf{CL}$};
\end{axis}
\end{tikzpicture}

(a) Linear scale
\end{minipage}%
\hspace{-1em}%
\begin{minipage}[t]{0.49\textwidth}
\centering
\begin{tikzpicture}
\begin{axis}[
  width=\linewidth,
  height=0.72\linewidth,
  xmode=log,
  log basis x=10,
  xmin=0.005,
  xmax=300,
  xtick={0.01,0.1,1,10,100},
  xticklabels={0.01,0.1,1,10,100},
  ymin=0,
  ymax=240,
  xlabel={Running time (s)},
  grid=major
]
\addplot[thick, const plot, mark=none,green!50!black] table[x=time,y=solved] {Figures/random_cactus_TOT.dat};
\addplot[thick, const plot, mark=none,red!70!black] table[x=time,y=solved] {Figures/random_cactus_TCT.dat};
\addplot[thick, const plot, mark=none,blue!70!black] table[x=time,y=solved] {Figures/random_cactus_CL.dat};
\node[font=\scriptsize,anchor=south east,text=green!50!black,yshift=1pt] at (axis cs:250,80) {$\mathsf{TOT}$};
\node[font=\scriptsize,anchor=south east,text=red!70!black,yshift=1pt] at (axis cs:250,118) {$\mathsf{TCT}$};
\node[font=\scriptsize,anchor=south east,text=blue!70!black,yshift=1pt] at (axis cs:250,200) {$\mathsf{CL}$};
\end{axis}
\end{tikzpicture}

(b) Logarithmic scale
\end{minipage}
\caption{Cumulative solved instances by running time.}
\label{fig:random-cactus}
\end{figure}
\vspace{-5mm}

The linear view in Fig.~\ref{fig:random-cactus}(a) shows $\mathsf{CL}$ continuing to solve instances over a broader range of running times and ultimately reaching a higher total, whereas the logarithmic view in Fig.~\ref{fig:random-cactus}(b) shows $\mathsf{TOT}$ and $\mathsf{TCT}$ initially accumulating solved instances faster.

\subsection{Query Statistics}

For brevity, we use MQ and EQ to denote system-level membership and equivalence queries, respectively.
To expose the per-instance trade-off, Fig.~\ref{fig:random-query-ratios} plots the ratios of the query counts of $\mathsf{CL}$ to those of each monolithic approach.

\vspace{-5mm}
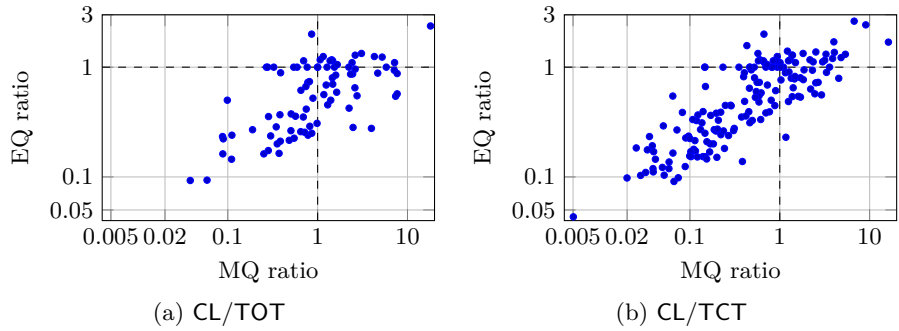
\begin{figure}[H]
\centering
\begin{minipage}[t]{0.49\textwidth}
\centering
\begin{tikzpicture}
\begin{axis}[
  width=\linewidth,
  height=0.72\linewidth,
  xmode=log,
  ymode=log,
  xmin=0.004,
  xmax=20,
  ymin=0.04,
  ymax=3,
  xtick={0.005,0.02,0.1,1,10},
  xticklabels={0.005,0.02,0.1,1,10},
  ytick={0.05,0.1,1,3},
  yticklabels={0.05,0.1,1,3},
  xlabel={MQ ratio},
  ylabel={EQ ratio},
  grid=major
]
\addplot+[only marks, mark=*, mark size=1.2pt] table[x=mq_ratio,y=eq_ratio] {Figures/random_query_ratios_CL_TOT.dat};
\addplot[black,dashed,domain=0.004:20,samples=2] {1};
\draw[black,dashed] (axis cs:1,0.04) -- (axis cs:1,3);
\end{axis}
\end{tikzpicture}

(a) $\mathsf{CL}/\mathsf{TOT}$
\end{minipage}%
\hfill
\begin{minipage}[t]{0.49\textwidth}
\centering
\begin{tikzpicture}
\begin{axis}[
  width=\linewidth,
  height=0.72\linewidth,
  xmode=log,
  ymode=log,
  xmin=0.004,
  xmax=20,
  ymin=0.04,
  ymax=3,
  xtick={0.005,0.02,0.1,1,10},
  xticklabels={0.005,0.02,0.1,1,10},
  ytick={0.05,0.1,1,3},
  yticklabels={0.05,0.1,1,3},
  xlabel={MQ ratio},
  ylabel={EQ ratio},
  grid=major
]
\addplot+[only marks, mark=*, mark size=1.2pt] table[x=mq_ratio,y=eq_ratio] {Figures/random_query_ratios_CL_TCT.dat};
\addplot[black,dashed,domain=0.004:20,samples=2] {1};
\draw[black,dashed] (axis cs:1,0.04) -- (axis cs:1,3);
\end{axis}
\end{tikzpicture}

(b) $\mathsf{CL}/\mathsf{TCT}$
\end{minipage}
\vspace{-1mm}
\caption{Per-instance query-count ratios.}
\label{fig:random-query-ratios}
\end{figure}
\vspace{-5mm}

For each jointly solved instance, the horizontal and vertical coordinates give the ratios of the MQ and EQ counts of $\mathsf{CL}$ to those of the corresponding monolithic approach, respectively.
Thus, points farther left indicate relatively fewer MQs for $\mathsf{CL}$, whereas points farther down indicate relatively fewer EQs.
The dashed lines mark equal query counts.
Both axes use logarithmic scales.
Each point represents an instance solved by both approaches, giving 85 points in Fig.~\ref{fig:random-query-ratios}(a) and 162 in Fig.~\ref{fig:random-query-ratios}(b).

$\mathsf{CL}$ uses no more EQs than $\mathsf{TOT}$ and $\mathsf{TCT}$ on 70 of 85 and 129 of 162 jointly solved instances, respectively.
The corresponding numbers for MQs are 47 and 114, while 45 and 103 instances satisfy both comparisons.
Thus, $\mathsf{CL}$'s relative reduction is more consistent for EQs than MQs.
This aligns with the algorithm structure, where the adapter submits an EQ only after collecting all component hypotheses, while resolving unknown TOTU entries may require additional MQs.
Possible spurious counterexamples may also trigger backtracking.
Overall, compositional learning changes the balance between query types rather than uniformly reducing both.

%% file: Tables/PARSensitivity.tex
\begin{table}[H]
\centering
\caption{Sensitivity to timeout penalties.}
\vspace{-1mm}
\label{tab:random-par-sensitivity}
\setlength{\tabcolsep}{12pt}
\begin{tabular}{crrrr}
\toprule
Approach & Solved & PAR-1 & PAR-2 & PAR-10 \\
\midrule
$\mathsf{TOT}$ & 89/240  & 196.013 & 384.763 & 1894.763 \\
$\mathsf{TCT}$ & 175/240 &  96.422 & 177.672 &  827.672 \\
$\mathsf{CL}$  & 213/240 &  54.423 &  88.173 &  358.173 \\
\bottomrule
\end{tabular}
\end{table}

%% file: Sections/Appendix/Illustrative_Example.tex
\section{Illustrative Example}
\label{app:illustrative-example}
\raggedbottom
\setlength{\intextsep}{2mm}

Consider the target NRTA $\mathcal{N}=(\mathcal{A}_1,\mathcal{A}_2)$ shown in Fig.~\ref{fig:illustrative-target} with $\Sigma_1=\{a,b\}$ and $\Sigma_2=\{b,c\}$.
The complete run uses $75$ system-level membership queries and $10$ system-level equivalence queries.

\input{Figures/IllustrativeTarget}

Let $\mathbf{T}_i^j$ denote a snapshot of the prepared TOTU of learner $\mathtt{L}_i$ at the $j$-th system-level equivalence query, and $\mathcal{H}_i^j$ the corresponding hypothesis.
The TOTUs and component hypotheses for rounds 1--5 are shown after their corresponding descriptions.
The TOTUs and hypotheses in rounds 6--10 are shown separately.
In each TOTU, the rows above the horizontal division are indexed by $\bm{S}$, and those below it are indexed by $\bm{R}$.
An entry with value $\unk$ is shown as ``$?(\beta)$'', where $\beta\in\{\top,\bot\}$ is its concrete value in the satisfying model used to construct the hypothesis.
Each system-level counterexample is shown to the right of the corresponding TOTUs, followed by its two component projections on separate lines.

\newif\ifillustrativeearly
\newif\ifillustrativeroundfour
\newif\ifillustrativehypfive
\newif\ifillustrativetotufour
\newif\ifillustrativetotufive
\newif\ifillustrativefive
\newif\ifillustrativesix
\newif\ifillustrativeseven
\newif\ifillustrativeeight
\newif\ifillustrativenine
\tikzset{
  hyp/.style={->,>=Stealth,initial text={},state/.append style={minimum size=4.5mm,inner sep=1pt}},
}

\smallskip
\noindent\textbf{Initialization.}
Each learner $\mathtt{L}_i$ initializes its TOTU with $\bm{S}=\bm{E}=\{\epsilon\}$ and $\bm{R}=\{([0,0],\sigma)\mid\sigma\in\Sigma_i\}$, where $f(\epsilon,\epsilon)=\top$ and all other entries have value $\unk$.
The adapter initializes each $\mathcal{P}_i$ with only its initial location and no transitions, and each backtracking stack $\mathsf{BT}_i$ is empty.

\smallskip
\noindent\textbf{Round 1.}
The learners fill and prepare $\mathbf{T}_1^1$ and $\mathbf{T}_2^1$ and construct the corresponding hypotheses.
The first equivalence query returns $u_{\mathit{ce}}^1=(9/2,a)\notin\mathcal{L}(\mathcal{N})$.
Its first projection has RTA region word $((4,5),a)$, while its second projection is $\epsilon$.
Since $\project{u_{\mathit{ce}}^1}{1}\notin\mathcal{L}(\mathcal{P}_1)$, the adapter returns $(\project{u_{\mathit{ce}}^1}{1},\bot)$ to $\mathtt{L}_1$.
Since the last action $a$ does not belong to $\Sigma_2$, it returns $(\epsilon,\top)$ to $\mathtt{L}_2$.
Learner $\mathtt{L}_1$ pushes $(\rabs{\project{u_{\mathit{ce}}^1}{1}},\mathbf{T}_1^1)$ onto $\mathsf{BT}_1$ and processes the $\bot$ feedback.
Figures~\ref{fig:illustrative-totu-round-1} and~\ref{fig:illustrative-hypotheses-one} show the corresponding TOTUs and component hypotheses.

\input{Tables/IllustrativeTOTURound1}
\input{Tables/IllustrativeHypothesisRound1}

\noindent\textbf{Round 2.}
The negative counterexample $u_{\mathit{ce}}^2=(0,a)(0,a)$ concerns only the first component.
Since $\project{u_{\mathit{ce}}^2}{1}\notin\mathcal{L}(\mathcal{P}_1)$, the adapter returns $(\project{u_{\mathit{ce}}^2}{1},\bot)$ to $\mathtt{L}_1$.
Since the last action $a$ does not belong to $\Sigma_2$, it returns $(\epsilon,\top)$ to $\mathtt{L}_2$.
Learner $\mathtt{L}_1$ pushes $(\rabs{\project{u_{\mathit{ce}}^2}{1}},\mathbf{T}_1^2)$ onto $\mathsf{BT}_1$ before processing the $\bot$ feedback.
Indeed, the first $a$ moves $\mathcal{A}_1$ from $q_0^1$ to $q_1^1$, where the second $a$ cannot occur.
Figures~\ref{fig:illustrative-totu-round-2} and~\ref{fig:illustrative-hypotheses-two} show the corresponding TOTUs and component hypotheses.

\input{Tables/IllustrativeTOTURound2}
\input{Tables/IllustrativeHypothesisRound2}

\smallskip
\noindent\textbf{Round 3.}
For $u_{\mathit{ce}}^3=(1/2,a)(2,b)\notin\mathcal{L}(\mathcal{N})$, the last action $b$ belongs to both component alphabets, and neither projection belongs to its positive-evidence RTA.
The adapter therefore returns $(\project{u_{\mathit{ce}}^3}{i},\bot)$ to each learner $\mathtt{L}_i$, which pushes $(\rabs{\project{u_{\mathit{ce}}^3}{i}},\mathbf{T}_i^3)$ onto $\mathsf{BT}_i$ before processing the $\bot$ feedback.
In the target system, $\project{u_{\mathit{ce}}^3}{1}\in\mathcal{L}(\mathcal{A}_1)$, whereas the accumulated delay of the projected $b$ is $5/2$ and hence $\project{u_{\mathit{ce}}^3}{2}\notin\mathcal{L}(\mathcal{A}_2)$.
Thus, the component-level counterexample returned to $\mathtt{L}_1$ is spurious, whereas that returned to $\mathtt{L}_2$ is not.
Figures~\ref{fig:illustrative-totu-round-3} and~\ref{fig:illustrative-hypotheses-three} show the corresponding TOTUs and component hypotheses.

\input{Tables/IllustrativeTOTURound3}
\input{Tables/IllustrativeHypothesisRound3}

\smallskip
\noindent\textbf{Round 4.}
The teacher returns the positive counterexample $u_{\mathit{ce}}^4=(0,a)(2,b)$.
The adapter calls $\mathsf{addPE}_1$ and $\mathsf{addPE}_2$ on the two projections, whose RTA region words are $([0,0],a)([2,2],b)$ and $([2,2],b)$, respectively, and returns $(\project{u_{\mathit{ce}}^4}{i},\top)$ to each learner $\mathtt{L}_i$.
For each $i\in\{1,2\}$, $\mathsf{BT}_i$ contains no $(\pi,\mathbf{T}_i^k)$ such that $\pi$ is a prefix of the corresponding RTA region word, so both learners process the $\top$ feedback without rollback.
Figures~\ref{fig:illustrative-totu-round-4} and~\ref{fig:illustrative-hypotheses-four} show the corresponding TOTUs and component hypotheses.

\illustrativetotufourtrue
\illustrativetotufivefalse
\input{Tables/IllustrativeTOTURoundsThreeFour}

\illustrativeearlyfalse
\illustrativeroundfourtrue
\illustrativehypfivefalse
\input{Tables/IllustrativeHypotheses}
\illustrativeroundfourfalse

\newpage
\noindent\textbf{Round 5.}
The teacher returns the negative counterexample
\[
  u_{\mathit{ce}}^5=(0,a)(0,b)(0,b).
\]
Since the last action $b$ belongs to both component alphabets and neither projection belongs to its positive-evidence RTA, the adapter returns $(\project{u_{\mathit{ce}}^5}{i},\bot)$ to each learner $\mathtt{L}_i$.
Each learner $\mathtt{L}_i$ pushes $(\rabs{\project{u_{\mathit{ce}}^5}{i}},\mathbf{T}_i^5)$ onto $\mathsf{BT}_i$ and processes the $\bot$ feedback.
In fact, its first projection belongs to $\mathcal{L}(\mathcal{A}_1)$, while its second projection does not belong to $\mathcal{L}(\mathcal{A}_2)$.
Figures~\ref{fig:illustrative-totu-round-5} and~\ref{fig:illustrative-hypotheses-five} show the corresponding TOTUs and component hypotheses.

\illustrativetotufourfalse
\illustrativetotufivetrue
\input{Tables/IllustrativeTOTURoundsThreeFour}
\illustrativetotufivefalse

\illustrativeearlyfalse
\illustrativeroundfourfalse
\illustrativehypfivetrue
\input{Tables/IllustrativeHypotheses}
\illustrativehypfivefalse

\newpage
\noindent\textbf{Round 6.}
The positive counterexample
\[
  u_{\mathit{ce}}^6=(0,a)(0,b)(0,c)(0,b)
\]
reveals the spurious feedback from round~5.
The adapter adds both projections to the corresponding positive-evidence RTAs and returns $(\project{u_{\mathit{ce}}^6}{i},\top)$ to each learner $\mathtt{L}_i$.
Its first projection has RTA region word
\[
  \pi=([0,0],a)([0,0],b)([0,0],b),
\]
and $(\pi,\mathbf{T}_1^5)$ was pushed onto $\mathsf{BT}_1$ in round~5.
Learner $\mathtt{L}_1$ restores $\mathbf{T}_1^5$, pops $(\pi,\mathbf{T}_1^5)$ and all entries above it, and then processes the $\top$ feedback, which in particular sets $f(\pi,\epsilon)=\top$.
The entries pushed in rounds 1--3 remain in $\mathsf{BT}_1$.

For the second component,
\[
  \rabs{\project{u_{\mathit{ce}}^6}{2}}
  =([0,0],b)([0,0],c)([0,0],b).
\]
$\mathsf{BT}_2$ contains no $(\pi',\mathbf{T}_2^k)$ such that $\pi'$ is a prefix of this RTA region word, so $\mathtt{L}_2$ processes the $\top$ feedback without rollback.
For every positive projection in the remaining rounds, the corresponding backtracking stack contains no entry whose first element is its prefix, so no further rollback occurs.

\input{Tables/IllustrativeTOTURound5}

\illustrativeearlyfalse
\illustrativeroundfourfalse
\illustrativefivetrue
\input{Tables/IllustrativeHypotheses}

\newpage
\noindent\textbf{Round 7.}
The positive counterexample $u_{\mathit{ce}}^7=(0,a)(0,b)(0,c)(2,b)$ further enriches both positive-evidence RTAs.
The adapter returns each projection with label $\top$ to its learner.
Both learners then process the $\top$ feedback.

\input{Tables/IllustrativeTOTURound6}

\illustrativefivefalse
\illustrativesixtrue
\illustrativesevenfalse
\illustrativeeightfalse
\illustrativeninefalse
\input{Tables/IllustrativeHypotheses}

\newpage
\noindent\textbf{Round 8.}
The next positive counterexample is
\[
  u_{\mathit{ce}}^8=(0,a)(2,b)(0,c)(2,b).
\]
The adapter again adds both projections to the positive-evidence RTAs and returns $(\project{u_{\mathit{ce}}^8}{i},\top)$ to each learner $\mathtt{L}_i$.

\input{Tables/IllustrativeTOTURound7}

\illustrativesixfalse
\illustrativeseventrue
\input{Tables/IllustrativeHypotheses}

\clearpage
\smallskip
\noindent\textbf{Round 9.}
The final counterexample $u_{\mathit{ce}}^9=(1/3,a)(1/12,b)(1/12,c)(23/12,b)$ is positive.
The adapter adds its projections to the positive-evidence RTAs and returns $(\project{u_{\mathit{ce}}^9}{i},\top)$ to each learner $\mathtt{L}_i$.

\input{Tables/IllustrativeTOTURound8}

\illustrativesevenfalse
\illustrativeeighttrue
\input{Tables/IllustrativeHypotheses}

\clearpage
\smallskip
\noindent\textbf{Round 10.}
The hypotheses constructed from $\mathbf{T}_1^{10}$ and $\mathbf{T}_2^{10}$ pass the final system-level equivalence query.

\input{Tables/IllustrativeTOTURound9}

\illustrativeeightfalse
\illustrativeninetrue
\input{Tables/IllustrativeHypotheses}

The final hypotheses do not recover the target components exactly.
In particular, $\mathcal{L}(\mathcal{H}_1^{10})\subsetneq\mathcal{L}(\mathcal{A}_1)$, although the learned and target networks have the same timed language.

%% file: Figures/IllustrativeTarget.tex
\begingroup
\setlength{\intextsep}{2mm}
\tikzset{
  hyp/.style={->,>=Stealth,initial text={},state/.append style={minimum size=4.5mm,inner sep=1pt}},
}
\begin{figure}[H]
\centering
\setlength{\tabcolsep}{5pt}
\begin{tabular}{c@{\hspace{8mm}}c}
\begin{tikzpicture}[hyp,baseline=(current bounding box.center)]
  \node[state,initial] (q0) {$q_0^1$};
  \node[state,right=25mm of q0] (q1) {$q_1^1$};
  \path (q0) edge node[above] {$a,\ 0\le x_1\le4$} (q1)
        (q1) edge[loop right] node[right] {$b$} (q1);
\end{tikzpicture}
&
\begin{tikzpicture}[hyp,baseline=(current bounding box.center)]
  \node[state,initial] (q0) {$q_0^2$};
  \node[state,right=22mm of q0] (q1) {$q_1^2$};
  \path (q0) edge[bend left=15] node[above] {$b,\ 0\le x_2\le2$} (q1)
        (q1) edge[bend left=15] node[below] {$c$} (q0);
\end{tikzpicture}
\\[3mm]
$\mathcal{A}_1$ & $\mathcal{A}_2$
\end{tabular}
\caption{Target components of the illustrative example}
\label{fig:illustrative-target}
\end{figure}
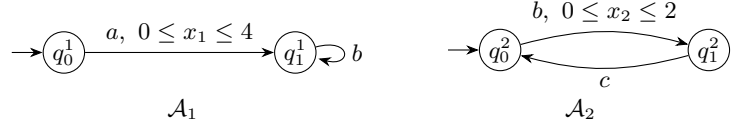
\endgroup

%% file: Tables/IllustrativeTOTURound1.tex
\begin{figure}[H]
\centering
\setlength{\tabcolsep}{5pt}
\resizebox{0.7\linewidth}{!}{%
\begin{tabular}{@{}l@{\hspace{3mm}}l@{\hspace{4mm}}l@{}}
\raisebox{-\height}{%
\begin{tabular}[t]{@{}l|c@{}}
\toprule
$\mathbf{T}_1^1$ & $\epsilon$\\ \midrule
$\epsilon$ & $\top$\\
\midrule
$([0,0],a)$ & $\top$\\
$([0,0],b)$ & $?(\top)$\\
\bottomrule
\end{tabular}%
} &
\raisebox{-\height}{%
\begin{tabular}[t]{@{}l|c@{}}
\toprule
$\mathbf{T}_2^1$ & $\epsilon$\\ \midrule
$\epsilon$ & $\top$\\
$([0,0],c)$ & $\bot$\\
\midrule
$([0,0],b)$ & $?(\bot)$\\
\bottomrule
\end{tabular}%
} &
\raisebox{\dimexpr-\height-3ex\relax}{%
\begin{tabular}[t]{@{}l@{}}
\begingroup
\setbox0=\hbox{$u_{\mathit{ce}}^1=(9/2,a)\notin\mathcal{L}(\mathcal{N})$}%
\setbox1=\hbox{$\rabs{\project{u_{\mathit{ce}}^1}{1}}=((4,5),a)$}%
\setbox2=\hbox{$\rabs{\project{u_{\mathit{ce}}^1}{2}}=\epsilon$}%
\dimen0=\wd0
\ifdim\wd1>\dimen0 \dimen0=\wd1\fi
\ifdim\wd2>\dimen0 \dimen0=\wd2\fi
\begin{tabular}[t]{@{}l@{}}
$u_{\mathit{ce}}^1=(9/2,a)\notin\mathcal{L}(\mathcal{N})$\\[-1ex]
$\xrightarrow{\makebox[\dimen0][c]{}}$\\[-0.5ex]
$\rabs{\project{u_{\mathit{ce}}^1}{1}}=((4,5),a)$\\[0.5ex]
$\rabs{\project{u_{\mathit{ce}}^1}{2}}=\epsilon$
\end{tabular}
\endgroup
\end{tabular}%
}
\end{tabular}%
}
\caption{TOTUs and counterexample in round 1}
\label{fig:illustrative-totu-round-1}
\end{figure}

%% file: Tables/IllustrativeHypothesisRound1.tex
\begin{figure}[H]
\centering
\setlength{\tabcolsep}{5pt}
\begin{tabular}{c@{\hspace{6mm}}c}
\begin{tikzpicture}[hyp,baseline=(current bounding box.center)]
  \node[state,initial] (q0) {$q_0$};
  \path (q0) edge[loop right] node[right,align=left] {$a$\\[3mm]$b$} (q0);
\end{tikzpicture}
&
\begin{tikzpicture}[hyp,baseline=(current bounding box.center)]
  \node[state,initial] (q0) {$q_0$};
\end{tikzpicture}
\\[3mm]
$\mathcal{H}_1^1$ & $\mathcal{H}_2^1$
\end{tabular}
\caption{Component hypotheses in round 1}
\label{fig:illustrative-hypotheses-one}
\end{figure}

%% file: Tables/IllustrativeTOTURound2.tex
\begin{figure}[H]
\centering
\setlength{\tabcolsep}{5pt}
\resizebox{0.7\linewidth}{!}{%
\begin{tabular}{@{}l@{\hspace{3mm}}l@{\hspace{4mm}}l@{}}
\raisebox{-\height}{%
\begin{tabular}[t]{@{}l|c@{}}
\toprule
$\mathbf{T}_1^2$ & $\epsilon$\\ \midrule
$\epsilon$ & $\top$\\
$((4,5),a)$ & $\bot$\\
\midrule
$([0,0],a)$ & $\top$\\
$([0,0],b)$ & $?(\bot)$\\
\bottomrule
\end{tabular}%
} &
\raisebox{-\height}{%
\begin{tabular}[t]{@{}l|c@{}}
\toprule
$\mathbf{T}_2^2$ & $\epsilon$\\ \midrule
$\epsilon$ & $\top$\\
$([0,0],c)$ & $\bot$\\
\midrule
$([0,0],b)$ & $?(\bot)$\\
\bottomrule
\end{tabular}%
} &
\raisebox{\dimexpr-\height-3ex\relax}{%
\begin{tabular}[t]{@{}l@{}}
\begingroup
\setbox0=\hbox{$u_{\mathit{ce}}^2=(0,a)(0,a)\notin\mathcal{L}(\mathcal{N})$}%
\setbox1=\hbox{$\rabs{\project{u_{\mathit{ce}}^2}{1}}=([0,0],a)([0,0],a)$}%
\setbox2=\hbox{$\rabs{\project{u_{\mathit{ce}}^2}{2}}=\epsilon$}%
\dimen0=\wd0
\ifdim\wd1>\dimen0 \dimen0=\wd1\fi
\ifdim\wd2>\dimen0 \dimen0=\wd2\fi
\begin{tabular}[t]{@{}l@{}}
$u_{\mathit{ce}}^2=(0,a)(0,a)\notin\mathcal{L}(\mathcal{N})$\\[-1ex]
$\xrightarrow{\makebox[\dimen0][c]{}}$\\[-0.5ex]
$\rabs{\project{u_{\mathit{ce}}^2}{1}}=([0,0],a)([0,0],a)$\\[0.5ex]
$\rabs{\project{u_{\mathit{ce}}^2}{2}}=\epsilon$
\end{tabular}
\endgroup
\end{tabular}%
}
\end{tabular}%
}
\caption{TOTUs and counterexample in round 2}
\label{fig:illustrative-totu-round-2}
\end{figure}

%% file: Tables/IllustrativeHypothesisRound2.tex
\begin{figure}[H]
\centering
\setlength{\tabcolsep}{5pt}
\begin{tabular}{c@{\hspace{6mm}}c}
\begin{tikzpicture}[hyp,baseline=(current bounding box.center)]
  \node[state,initial] (q0) {$q_0$};
  \path (q0) edge[loop right] node[right] {$a,\ 0\le x_1\le4$} (q0);
\end{tikzpicture}
&
\begin{tikzpicture}[hyp,baseline=(current bounding box.center)]
  \node[state,initial] (q0) {$q_0$};
\end{tikzpicture}
\\[3mm]
$\mathcal{H}_1^2$ & $\mathcal{H}_2^2$
\end{tabular}
\caption{Component hypotheses in round 2}
\label{fig:illustrative-hypotheses-two}
\end{figure}

%% file: Tables/IllustrativeTOTURound3.tex
\begin{figure}[H]
\centering
\setlength{\tabcolsep}{5pt}
\resizebox{0.9\linewidth}{!}{%
\begin{tabular}{@{}l@{\hspace{3mm}}l@{\hspace{4mm}}l@{}}
\raisebox{-\height}{%
\begin{tabular}[t]{@{}l|cc@{}}
\toprule
$\mathbf{T}_1^3$ & $\epsilon$ & $([0,0],a)$\\ \midrule
$\epsilon$ & $\top$ & $\top$\\
$((4,5),a)$ & $\bot$ & $\bot$\\
$([0,0],a)$ & $\top$ & $\bot$\\
\midrule
$([0,0],b)$ & $?(\bot)$ & $?(\bot)$\\
$([0,0],a)([0,0],a)$ & $\bot$ & $\bot$\\
$([0,0],a)([0,0],b)$ & $\top$ & $\bot$\\
\bottomrule
\end{tabular}%
} &
\raisebox{-\height}{%
\begin{tabular}[t]{@{}l|c@{}}
\toprule
$\mathbf{T}_2^3$ & $\epsilon$\\ \midrule
$\epsilon$ & $\top$\\
$([0,0],c)$ & $\bot$\\
\midrule
$([0,0],b)$ & $\top$\\
\bottomrule
\end{tabular}%
} &
\raisebox{\dimexpr-\height-3ex\relax}{%
\begin{tabular}[t]{@{}l@{}}
\begingroup
\setbox0=\hbox{$u_{\mathit{ce}}^3=(1/2,a)(2,b)\notin\mathcal{L}(\mathcal{N})$}%
\setbox1=\hbox{$\rabs{\project{u_{\mathit{ce}}^3}{1}}=((0,1),a)([2,2],b)$}%
\setbox2=\hbox{$\rabs{\project{u_{\mathit{ce}}^3}{2}}=((2,3),b)$}%
\dimen0=\wd0
\ifdim\wd1>\dimen0 \dimen0=\wd1\fi
\ifdim\wd2>\dimen0 \dimen0=\wd2\fi
\begin{tabular}[t]{@{}l@{}}
$u_{\mathit{ce}}^3=(1/2,a)(2,b)\notin\mathcal{L}(\mathcal{N})$\\[-1ex]
$\xrightarrow{\makebox[\dimen0][c]{}}$\\[-0.5ex]
$\rabs{\project{u_{\mathit{ce}}^3}{1}}=((0,1),a)([2,2],b)$\\[0.5ex]
$\rabs{\project{u_{\mathit{ce}}^3}{2}}=((2,3),b)$
\end{tabular}
\endgroup
\end{tabular}%
}
\end{tabular}%
}
\caption{TOTUs and counterexample in round 3}
\label{fig:illustrative-totu-round-3}
\end{figure}

%% file: Tables/IllustrativeHypothesisRound3.tex
\begin{figure}[H]
\centering
\setlength{\tabcolsep}{5pt}
\begin{tabular}{c@{\hspace{6mm}}c}
\begin{tikzpicture}[hyp,baseline=(current bounding box.center)]
  \node[state,initial] (q0) {$q_0$};
  \node[state,right=25mm of q0] (q1) {$q_1$};
  \path (q0) edge node[above] {$a,\ 0\le x_1\le4$} (q1)
        (q1) edge[loop right] node[right] {$b$} (q1);
\end{tikzpicture}
&
\begin{tikzpicture}[hyp,baseline=(current bounding box.center)]
  \node[state,initial] (q0) {$q_0$};
  \path (q0) edge[loop right] node[right] {$b$} (q0);
\end{tikzpicture}
\\[3mm]
$\mathcal{H}_1^3$ & $\mathcal{H}_2^3$
\end{tabular}
\caption{Component hypotheses in round 3}
\label{fig:illustrative-hypotheses-three}
\end{figure}

%% file: Tables/IllustrativeTOTURoundsThreeFour.tex
\ifillustrativetotufour
\begin{figure}[H]
\centering
\setlength{\tabcolsep}{2.5pt}
\resizebox{0.82\linewidth}{!}{%
\begin{tabular}{@{}l@{\hspace{3mm}}l@{\hspace{4mm}}l@{}}
\raisebox{-\height}{%
\begin{tabular}[t]{@{}l|cc@{}}
\toprule
$\mathbf{T}_1^4$ & $\epsilon$ & $([0,0],a)$\\ \midrule
$\epsilon$ & $\top$ & $\top$\\
$((4,5),a)$ & $\bot$ & $\bot$\\
$([0,0],a)$ & $\top$ & $\bot$\\
\midrule
$((0,1),a)$ & $\top$ & $\bot$\\
$([0,0],b)$ & $\bot$ & $\bot$\\
$((0,1),a)([2,2],b)$ & $\bot$ & $\bot$\\
$([0,0],a)([0,0],a)$ & $\bot$ & $\bot$\\
$([0,0],a)([0,0],b)$ & $\top$ & $\bot$\\
\bottomrule
\end{tabular}%
} &
\raisebox{-\height}{%
\begin{tabular}[t]{@{}l|c@{}}
\toprule
$\mathbf{T}_2^4$ & $\epsilon$\\ \midrule
$\epsilon$ & $\top$\\
$([0,0],c)$ & $\bot$\\
\midrule
$((2,3),b)$ & $\bot$\\
$([0,0],b)$ & $\top$\\
\bottomrule
\end{tabular}%
} &
\raisebox{\dimexpr-\height-3ex\relax}{%
\begin{tabular}[t]{@{}l@{}}
\begingroup
\setbox0=\hbox{$ u_{\mathit{ce}}^4=(0,a)(2,b)\in\mathcal{L}(\mathcal{N}) $}%
\setbox1=\hbox{$ \rabs{\project{u_{\mathit{ce}}^4}{1}}=([0,0],a)([2,2],b) $}%
\setbox2=\hbox{$ \rabs{\project{u_{\mathit{ce}}^4}{2}}=([2,2],b) $}%
\dimen0=\wd0
\ifdim\wd1>\dimen0 \dimen0=\wd1\fi
\ifdim\wd2>\dimen0 \dimen0=\wd2\fi
\begin{tabular}[t]{@{}l@{}}
$ u_{\mathit{ce}}^4=(0,a)(2,b)\in\mathcal{L}(\mathcal{N}) $\\[-1ex]
$\xrightarrow{\makebox[\dimen0][c]{}}$\\[-0.5ex]
$\rabs{\project{u_{\mathit{ce}}^4}{1}}=([0,0],a)([2,2],b)$\\[0.5ex]
$\rabs{\project{u_{\mathit{ce}}^4}{2}}=([2,2],b)$
\end{tabular}
\endgroup
\end{tabular}%
}
\end{tabular}%
}
\caption{TOTUs and counterexample in round 4}
\label{fig:illustrative-totu-round-4}
\end{figure}
\fi

\ifillustrativetotufive
\begin{figure}[H]
\centering
\setlength{\tabcolsep}{2.5pt}
\resizebox{0.85\linewidth}{!}{%
\begin{tabular}{@{}l@{\hspace{4mm}}l@{}}
\raisebox{-\height}{%
\begin{tabular}[t]{@{}l@{}}
\begin{tabular}[t]{@{}l|ccc@{}}
\toprule
$\mathbf{T}_1^5$ & $\epsilon$ & $([0,0],a)$ & $([2,2],b)$\\ \midrule
$\epsilon$ & $\top$ & $\top$ & $\bot$\\
$((0,1),a)$ & $\top$ & $\bot$ & $\bot$\\
$((4,5),a)$ & $\bot$ & $\bot$ & $\bot$\\
$([0,0],a)$ & $\top$ & $\bot$ & $\top$\\
\midrule
$([0,0],b)$ & $\bot$ & $\bot$ & $\bot$\\
$((0,1),a)([0,0],a)$ & $\bot$ & $\bot$ & $\bot$\\
$((0,1),a)([0,0],b)$ & $\top$ & $\bot$ & $?(\bot)$\\
$((0,1),a)([2,2],b)$ & $\bot$ & $\bot$ & $\bot$\\
$([0,0],a)([0,0],a)$ & $\bot$ & $\bot$ & $\bot$\\
$([0,0],a)([0,0],b)$ & $\top$ & $\bot$ & $?(\bot)$\\
$([0,0],a)([2,2],b)$ & $\top$ & $\bot$ & $?(\bot)$\\
\bottomrule
\end{tabular}%
\end{tabular}%
} &
\raisebox{-\height}{%
\begin{tabular}[t]{@{}l@{}}
\begin{tabular}[t]{@{}l|c@{}}
\toprule
$\mathbf{T}_2^5$ & $\epsilon$\\ \midrule
$\epsilon$ & $\top$\\
$([0,0],c)$ & $\bot$\\
\midrule
$((2,3),b)$ & $\bot$\\
$([0,0],b)$ & $\top$\\
$([2,2],b)$ & $\top$\\
\bottomrule
\end{tabular}\\[35mm]
\begin{tabular}[t]{@{}l@{}}
\begingroup
\setbox0=\hbox{$ u_{\mathit{ce}}^5=(0,a)(0,b)(0,b)\notin\mathcal{L}(\mathcal{N}) $}%
\setbox1=\hbox{$ \rabs{\project{u_{\mathit{ce}}^5}{1}}=([0,0],a)([0,0],b)([0,0],b) $}%
\setbox2=\hbox{$ \rabs{\project{u_{\mathit{ce}}^5}{2}}=([0,0],b)([0,0],b) $}%
\dimen0=\wd0
\ifdim\wd1>\dimen0 \dimen0=\wd1\fi
\ifdim\wd2>\dimen0 \dimen0=\wd2\fi
\begin{tabular}[t]{@{}l@{}}
$ u_{\mathit{ce}}^5=(0,a)(0,b)(0,b)\notin\mathcal{L}(\mathcal{N}) $\\[-1ex]
$\xrightarrow{\makebox[\dimen0][c]{}}$\\[-0.5ex]
$\rabs{\project{u_{\mathit{ce}}^5}{1}}=([0,0],a)([0,0],b)([0,0],b)$\\[0.5ex]
$\rabs{\project{u_{\mathit{ce}}^5}{2}}=([0,0],b)([0,0],b)$
\end{tabular}
\endgroup
\end{tabular}
\end{tabular}%
}
\end{tabular}
}
\caption{TOTUs and counterexample in round 5}
\label{fig:illustrative-totu-round-5}
\end{figure}
\fi

%% file: Tables/IllustrativeHypotheses.tex
\begingroup
\tikzset{
  hyp/.style={->,>=Stealth,initial text={},state/.append style={minimum size=4.5mm,inner sep=1pt}},
}

\ifillustrativeearly
\begin{figure}[H]
\centering
\setlength{\tabcolsep}{5pt}
\begin{tabular}{c@{\hspace{6mm}}c}
\begin{tikzpicture}[hyp,baseline=(current bounding box.center)]
  \node[state,initial] (q0) {$q_0$};
  \path (q0) edge[loop right] node[right,align=left] {$a$\\[3mm]$b$} (q0);
\end{tikzpicture}
&
\begin{tikzpicture}[hyp,baseline=(current bounding box.center)]
  \node[state,initial] (q0) {$q_0$};
\end{tikzpicture}
\\[3mm]
 $\mathcal{H}_1^1$ & $\mathcal{H}_2^1$
\end{tabular}
\caption{Component hypotheses in round 1}
\label{fig:illustrative-hypotheses-one}
\end{figure}

\begin{figure}[H]
\centering
\setlength{\tabcolsep}{5pt}
\begin{tabular}{c@{\hspace{6mm}}c}
\begin{tikzpicture}[hyp,baseline=(current bounding box.center)]
  \node[state,initial] (q0) {$q_0$};
  \path (q0) edge[loop right] node[right] {$a,\ 0\le x_1\le4$} (q0);
\end{tikzpicture}
&
\begin{tikzpicture}[hyp,baseline=(current bounding box.center)]
  \node[state,initial] (q0) {$q_0$};
\end{tikzpicture}
\\[3mm]
 $\mathcal{H}_1^2$ & $\mathcal{H}_2^2$
\end{tabular}
\caption{Component hypotheses in round 2}
\label{fig:illustrative-hypotheses-two}
\end{figure}

\begin{figure}[H]
\centering
\setlength{\tabcolsep}{5pt}
\begin{tabular}{c@{\hspace{6mm}}c}
\begin{tikzpicture}[hyp,baseline=(current bounding box.center)]
  \node[state,initial] (q0) {$q_0$};
  \node[state,right=25mm of q0] (q1) {$q_1$};
  \path (q0) edge node[above] {$a,\ 0\le x_1\le4$} (q1)
        (q1) edge[loop right] node[right] {$b$} (q1);
\end{tikzpicture}
&
\begin{tikzpicture}[hyp,baseline=(current bounding box.center)]
  \node[state,initial] (q0) {$q_0$};
  \path (q0) edge[loop right] node[right] {$b$} (q0);
\end{tikzpicture}
\\[3mm]
 $\mathcal{H}_1^3$ & $\mathcal{H}_2^3$
\end{tabular}
\caption{Component hypotheses in round 3}
\label{fig:illustrative-hypotheses-three}
\end{figure}
\fi

\ifillustrativeroundfour
\begin{figure}[H]
\centering
\setlength{\tabcolsep}{5pt}
\begin{tabular}{c@{\hspace{6mm}}c}
\begin{tikzpicture}[hyp,baseline=(current bounding box.center)]
  \node[state,initial] (q0) {$q_0$};
  \node[state,right=23.5mm of q0] (q1) {$q_1$};
  \path (q0) edge node[above] {$a,\ 0\le x_1\le4$} (q1)
        (q1) edge[loop right] node[right] {$b,\ 0\le x_1<2$} (q1);
\end{tikzpicture}
&
\begin{tikzpicture}[hyp,baseline=(current bounding box.center)]
  \node[state,initial] (q0) {$q_0$};
  \path (q0) edge[loop right] node[right] {$b,\ 0\le x_2\le2$} (q0);
\end{tikzpicture}
\\[3mm]

$\mathcal{H}_1^4$ & $\mathcal{H}_2^4$
\end{tabular}
\caption{Component hypotheses in round 4}
\label{fig:illustrative-hypotheses-four}
\end{figure}
\fi

\ifillustrativehypfive
\begin{figure}[H]
\centering
\setlength{\tabcolsep}{5pt}
\begin{tabular}{c@{\hspace{6mm}}c}

\begin{tikzpicture}[hyp,baseline=(current bounding box.center)]
  \node[state,initial] (q0) {$q_0$};
  \node[state,above right=2mm and 25mm of q0] (q1) {$q_1$};
  \node[state,below right=2mm and 25mm of q0] (q3) {$q_3$};
  \path (q0) edge node[above,sloped] {$a,\ x_1=0$} (q1)
             edge node[below,sloped] {$a,\ 0<x_1\le4$} (q3)
        (q1) edge node[right] {$b$} (q3)
        (q3) edge[loop right] node[right] {$b,\ 0\le x_1<2$} (q3);
\end{tikzpicture}
&
\begin{tikzpicture}[hyp,baseline=(current bounding box.center)]
  \node[state,initial] (q0) {$q_0$};
  \path (q0) edge[loop right] node[right] {$b,\ 0\le x_2\le2$} (q0);
\end{tikzpicture}
\\[3mm]
$\mathcal{H}_1^5$ & $\mathcal{H}_2^5$
\end{tabular}
\caption{Component hypotheses in round 5}
\label{fig:illustrative-hypotheses-five}
\end{figure}
\fi

\ifillustrativefive
\begin{figure}[H]
\centering
\setlength{\tabcolsep}{5pt}
\begin{tabular}{c@{\hspace{8mm}}c}
\begin{tikzpicture}[hyp,baseline=(current bounding box.center)]
  \node[state,initial] (q0) {$q_0$};
  \node[state,above right=3mm and 23mm of q0] (q1) {$q_1$};
  \node[state,below right=3mm and 23mm of q0] (q3) {$q_3$};
  \node[state,right=46mm of q0] (q4) {$q_4$};
  \path (q0) edge node[above,sloped] {$a,\ x_1=0$} (q1)
             edge node[below,sloped] {$a,\ 0<x_1\le4$} (q3)
        (q1) edge node[above,sloped] {$b$} (q4)
        (q3) edge node[below,sloped] {$b,\ 0\le x_1<2$} (q4);
\end{tikzpicture}
&
\begin{tikzpicture}[hyp,baseline=(current bounding box.center)]
  \node[state,initial] (q0) {$q_0$};
  \node[state,right=22mm of q0] (q1) {$q_1$};
  \path (q0) edge node[above] {$b,\ 0\le x_2\le2$} (q1)
        (q1) edge[loop right] node[right] {$c$} (q1);
\end{tikzpicture}
\\[3mm]
$\mathcal{H}_1^6$ & $\mathcal{H}_2^6$
\end{tabular}
\caption{Component hypotheses in round 6}
\label{fig:illustrative-hypotheses-six}
\end{figure}
\fi

\ifillustrativesix
\begin{figure}[H]
\centering
\setlength{\tabcolsep}{5pt}
\begin{tabular}{c@{\hspace{8mm}}c}
\begin{tikzpicture}[hyp,baseline=(current bounding box.center)]
  \node[state,initial] (q0) {$q_0$};
  \node[state,above right=3mm and 25mm of q0] (q1) {$q_1$};
  \node[state,below right=3mm and 25mm of q0] (q3) {$q_3$};
  \path (q0) edge node[above,sloped] {$a,\ x_1=0$} (q1)
             edge node[below,sloped] {$a,\ 0<x_1\le4$} (q3)
        (q1) edge node[right] {$b$} (q3)
        (q3) edge[loop right] node[right] {$b,\ 0\le x_1<2$} (q3);
\end{tikzpicture}
&
\begin{tikzpicture}[hyp,baseline=(current bounding box.center)]
  \node[state,initial] (q0) {$q_0$};
  \node[state,right=22mm of q0] (q1) {$q_1$};
  \path (q0) edge[bend left=15] node[above] {$b,\ 0\le x_2\le2$} (q1)
        (q1) edge[bend left=15] node[below] {$c$} (q0);
\end{tikzpicture}
\\[3mm]
 $\mathcal{H}_1^7$ & $\mathcal{H}_2^7$
\end{tabular}
\caption{Component hypotheses in round 7}
\label{fig:illustrative-hypotheses-seven}
\end{figure}
\fi

\ifillustrativeseven
\begin{figure}[H]
\centering
\setlength{\tabcolsep}{5pt}
\begin{tabular}{c@{\hspace{8mm}}c}
\begin{tikzpicture}[hyp,baseline=(current bounding box.center)]
  \node[state,initial] (q0) {$q_0$};
  \node[state,above right=3mm and 25mm of q0] (q1) {$q_1$};
  \node[state,below right=3mm and 25mm of q0] (q3) {$q_3$};
  \path (q0) edge node[above,sloped] {$a,\ x_1=0$} (q1)
             edge node[below,sloped] {$a,\ 0<x_1\le4$} (q3)
        (q1) edge[loop right] node[right] {$b,\ 0\le x_1<2$} (q1)
             edge node[right] {$b,\ 2\le x_1$} (q3)
        (q3) edge[loop right] node[right] {$b,\ 0\le x_1<2$} (q3);
\end{tikzpicture}
&
\begin{tikzpicture}[hyp,baseline=(current bounding box.center)]
  \node[state,initial] (q0) {$q_0$};
  \node[state,right=22mm of q0] (q1) {$q_1$};
  \path (q0) edge[bend left=15] node[above] {$b,\ 0\le x_2\le2$} (q1)
        (q1) edge[bend left=15] node[below] {$c$} (q0);
\end{tikzpicture}
\\[3mm]
 $\mathcal{H}_1^8$ & $\mathcal{H}_2^8$
\end{tabular}
\caption{Component hypotheses in round 8}
\label{fig:illustrative-hypotheses-eight}
\end{figure}
\fi

\ifillustrativeeight
\begin{figure}[H]
\centering
\setlength{\tabcolsep}{5pt}
\begin{tabular}{c@{\hspace{8mm}}c}
\begin{tikzpicture}[hyp,baseline=(current bounding box.center)]
  \node[state,initial] (q0) {$q_0$};
  \node[state,above right=3mm and 25mm of q0] (q1) {$q_1$};
  \node[state,below right=3mm and 25mm of q0] (q3) {$q_3$};
  \path (q0) edge node[above,sloped] {$a,\ x_1=0$} (q1)
             edge node[below,sloped] {$a,\ 0<x_1\le4$} (q3)
        (q1) edge[loop right] node[right] {$b$} (q1)
        (q3) edge[loop right] node[right] {$b,\ 0\le x_1<2$} (q3);
\end{tikzpicture}
&
\begin{tikzpicture}[hyp,baseline=(current bounding box.center)]
  \node[state,initial] (q0) {$q_0$};
  \node[state,right=22mm of q0] (q1) {$q_1$};
  \path (q0) edge[bend left=15] node[above] {$b,\ 0\le x_2\le2$} (q1)
        (q1) edge[bend left=15] node[below] {$c$} (q0);
\end{tikzpicture}
\\[3mm]
 $\mathcal{H}_1^9$ & $\mathcal{H}_2^9$
\end{tabular}
\caption{Component hypotheses in round 9}
\label{fig:illustrative-hypotheses-nine}
\end{figure}
\fi

\ifillustrativenine
\begin{figure}[H]
\centering
\setlength{\tabcolsep}{5pt}
\begin{tabular}{c@{\hspace{8mm}}c}
\begin{tikzpicture}[hyp,baseline=(current bounding box.center)]
  \node[state,initial] (q0) {$q_0$};
  \node[state,above right=3mm and 25mm of q0] (q1) {$q_1$};
  \node[state,below right=3mm and 25mm of q0] (q3) {$q_3$};
  \path (q0) edge node[above,sloped] {$a,\ x_1=0$} (q1)
             edge node[below,sloped] {$a,\ 0<x_1\le4$} (q3)
        (q1) edge[loop right] node[right] {$b$} (q1)
        (q3) edge node[right] {$b,\ 0\le x_1<2$} (q1);
\end{tikzpicture}
&
\begin{tikzpicture}[hyp,baseline=(current bounding box.center)]
  \node[state,initial] (q0) {$q_0$};
  \node[state,right=22mm of q0] (q1) {$q_1$};
  \path (q0) edge[bend left=15] node[above] {$b,\ 0\le x_2\le2$} (q1)
        (q1) edge[bend left=15] node[below] {$c$} (q0);
\end{tikzpicture}
\\[3mm]
 $\mathcal{H}_1^{10}$ & $\mathcal{H}_2^{10}$
\end{tabular}
\caption{Component hypotheses in round 10}
\label{fig:illustrative-hypotheses-ten}
\end{figure}
\fi
\endgroup

%% file: Tables/IllustrativeTOTURound5.tex
\begin{figure}[H]
\centering
\scriptsize
\setlength{\tabcolsep}{2pt}
\resizebox{0.98\linewidth}{!}{%
\begin{tabular}{@{}l@{\hspace{1mm}}l@{}}
\raisebox{-\height}{%
\begin{tabular}[t]{@{}l@{}}
\begin{tabular}[t]{@{}l|cccc@{}}
\toprule
$\mathbf{T}_1^6$ & $\epsilon$ & $([0,0],a)$ & $([2,2],b)$ & $([0,0],b)$\\ \midrule
$\epsilon$ & $\top$ & $\top$ & $\bot$ & $\bot$\\
$((0,1),a)$ & $\top$ & $\bot$ & $\bot$ & $\top$\\
$((4,5),a)$ & $\bot$ & $\bot$ & $\bot$ & $\bot$\\
$([0,0],a)$ & $\top$ & $\bot$ & $\top$ & $\top$\\
$([0,0],a)([0,0],b)$ & $\top$ & $\bot$ & $?(\bot)$ & $\bot$\\
\midrule
$([0,0],b)$ & $\bot$ & $\bot$ & $\bot$ & $\bot$\\
$((0,1),a)([0,0],a)$ & $\bot$ & $\bot$ & $\bot$ & $\bot$\\
$((0,1),a)([0,0],b)$ & $\top$ & $\bot$ & $?(\bot)$ & $?(\bot)$\\
$((0,1),a)([2,2],b)$ & $\bot$ & $\bot$ & $\bot$ & $\bot$\\
$([0,0],a)([0,0],a)$ & $\bot$ & $\bot$ & $\bot$ & $\bot$\\
$([0,0],a)([2,2],b)$ & $\top$ & $\bot$ & $?(\bot)$ & $?(\bot)$\\
$([0,0],a)([0,0],b)([0,0],a)$ & $\bot$ & $\bot$ & $\bot$ & $\bot$\\
$([0,0],a)([0,0],b)([0,0],b)$ & $\bot$ & $\bot$ & $\bot$ & $\bot$\\
\bottomrule
\end{tabular}%
\end{tabular}%
} &
\raisebox{-\height}{%
\begin{tabular}[t]{@{}l@{}}
\begin{tabular}[t]{@{}l|cc@{}}
\toprule
$\mathbf{T}_2^6$ & $\epsilon$ & $([0,0],b)$\\ \midrule
$\epsilon$ & $\top$ & $\top$\\
$([0,0],b)$ & $\top$ & $\bot$\\
$([0,0],c)$ & $\bot$ & $\bot$\\
\midrule
$((2,3),b)$ & $\bot$ & $\bot$\\
$([2,2],b)$ & $\top$ & $?(\bot)$\\
$([0,0],b)([0,0],b)$ & $\bot$ & $\bot$\\
$([0,0],b)([0,0],c)$ & $\top$ & $?(\bot)$\\
\bottomrule
\end{tabular}\\[32mm]
\begin{tabular}[t]{@{}l@{}}
\begingroup
\setbox0=\hbox{$ u_{\mathit{ce}}^6=(0,a)(0,b)(0,c)(0,b)\in\mathcal{L}(\mathcal{N}) $}%
\setbox1=\hbox{$ \rabs{\project{u_{\mathit{ce}}^6}{1}}=([0,0],a)([0,0],b)([0,0],b) $}%
\setbox2=\hbox{$ \rabs{\project{u_{\mathit{ce}}^6}{2}}=([0,0],b)([0,0],c)([0,0],b) $}%
\dimen0=\wd0
\ifdim\wd1>\dimen0 \dimen0=\wd1\fi
\ifdim\wd2>\dimen0 \dimen0=\wd2\fi
\begin{tabular}[t]{@{}l@{}}
$ u_{\mathit{ce}}^6=(0,a)(0,b)(0,c)(0,b)\in\mathcal{L}(\mathcal{N}) $\\[-1ex]
$\xrightarrow{\makebox[\dimen0][c]{}}$\\[-0.5ex]
$\rabs{\project{u_{\mathit{ce}}^6}{1}}=([0,0],a)([0,0],b)([0,0],b)$\\[0.5ex]
$\rabs{\project{u_{\mathit{ce}}^6}{2}}=([0,0],b)([0,0],c)([0,0],b)$
\end{tabular}
\endgroup
\end{tabular}
\end{tabular}
}
\end{tabular}%
}
    \caption{TOTUs and counterexample in round 6}
\label{fig:illustrative-totu-round-6}
\end{figure}

%% file: Tables/IllustrativeTOTURound6.tex
\begin{figure}[H]
\centering

\setlength{\tabcolsep}{2.5pt}
\resizebox{0.9\textwidth}{!}{%
\begin{tabular}{@{}l@{\hspace{4mm}}l@{}}
\raisebox{-\height}{%
\begin{tabular}[t]{@{}l@{}}
\begin{tabular}[t]{@{}l|ccc@{}}
\toprule
$\mathbf{T}_1^7$ & $\epsilon$ & $([0,0],a)$ & $([2,2],b)$\\ \midrule
$\epsilon$ & $\top$ & $\top$ & $\bot$\\
$((0,1),a)$ & $\top$ & $\bot$ & $\bot$\\
$((4,5),a)$ & $\bot$ & $\bot$ & $\bot$\\
$([0,0],a)$ & $\top$ & $\bot$ & $\top$\\
\midrule
$([0,0],b)$ & $\bot$ & $\bot$ & $\bot$\\
$((0,1),a)([0,0],a)$ & $\bot$ & $\bot$ & $\bot$\\
$((0,1),a)([0,0],b)$ & $\top$ & $\bot$ & $?(\bot)$\\
$((0,1),a)([2,2],b)$ & $\bot$ & $\bot$ & $\bot$\\
$([0,0],a)([0,0],a)$ & $\bot$ & $\bot$ & $\bot$\\
$([0,0],a)([0,0],b)$ & $\top$ & $\bot$ & $?(\bot)$\\
$([0,0],a)([2,2],b)$ & $\top$ & $\bot$ & $?(\bot)$\\
$([0,0],a)([0,0],b)([0,0],b)$ & $\top$ & $\bot$ & $?(\bot)$\\
\bottomrule
\end{tabular}%
\end{tabular}%
} &
\raisebox{-\height}{%
\begin{tabular}[t]{@{}l@{}}
\begin{tabular}[t]{@{}l|cc@{}}
\toprule
$\mathbf{T}_2^7$ & $\epsilon$ & $([0,0],b)$\\ \midrule
$\epsilon$ & $\top$ & $\top$\\
$([0,0],b)$ & $\top$ & $\bot$\\
$([0,0],c)$ & $\bot$ & $\bot$\\
\midrule
$((2,3),b)$ & $\bot$ & $\bot$\\
$([2,2],b)$ & $\top$ & $?(\bot)$\\
$([0,0],b)([0,0],b)$ & $\bot$ & $\bot$\\
$([0,0],b)([0,0],c)$ & $\top$ & $\top$\\
$([0,0],b)([0,0],c)([0,0],b)$ & $\top$ & $?(\bot)$\\
\bottomrule
\end{tabular}\\[41mm]
\begin{tabular}[t]{@{}l@{}}
\begingroup
\setbox0=\hbox{$ u_{\mathit{ce}}^7=(0,a)(0,b)(0,c)(2,b)\in\mathcal{L}(\mathcal{N}) $}%
\setbox1=\hbox{$ \rabs{\project{u_{\mathit{ce}}^7}{1}}=([0,0],a)([0,0],b)([2,2],b) $}%
\setbox2=\hbox{$ \rabs{\project{u_{\mathit{ce}}^7}{2}}=([0,0],b)([0,0],c)([2,2],b) $}%
\dimen0=\wd0
\ifdim\wd1>\dimen0 \dimen0=\wd1\fi
\ifdim\wd2>\dimen0 \dimen0=\wd2\fi
\begin{tabular}[t]{@{}l@{}}
$ u_{\mathit{ce}}^7=(0,a)(0,b)(0,c)(2,b)\in\mathcal{L}(\mathcal{N}) $\\[-1ex]
$\xrightarrow{\makebox[\dimen0][c]{}}$\\[-0.5ex]
$\rabs{\project{u_{\mathit{ce}}^7}{1}}=([0,0],a)([0,0],b)([2,2],b)$\\[0.5ex]
$\rabs{\project{u_{\mathit{ce}}^7}{2}}=([0,0],b)([0,0],c)([2,2],b)$
\end{tabular}
\endgroup
\end{tabular}
\end{tabular}%
}
\end{tabular}%
}
    \caption{TOTUs and counterexample in round 7}
\label{fig:illustrative-totu-round-7}
\end{figure}

%% file: Tables/IllustrativeTOTURound7.tex
\begin{figure}[H]
\centering

\setlength{\tabcolsep}{2.5pt}
\resizebox{0.93\textwidth}{!}{%
\begin{tabular}{@{}l@{\hspace{4mm}}l@{}}
\raisebox{-\height}{%
\begin{tabular}[t]{@{}l@{}}
\begin{tabular}[t]{@{}l|ccc@{}}
\toprule
$\mathbf{T}_1^8$ & $\epsilon$ & $([0,0],a)$ & $([2,2],b)$\\ \midrule
$\epsilon$ & $\top$ & $\top$ & $\bot$\\
$((0,1),a)$ & $\top$ & $\bot$ & $\bot$\\
$((4,5),a)$ & $\bot$ & $\bot$ & $\bot$\\
$([0,0],a)$ & $\top$ & $\bot$ & $\top$\\
\midrule
$([0,0],b)$ & $\bot$ & $\bot$ & $\bot$\\
$((0,1),a)([0,0],a)$ & $\bot$ & $\bot$ & $\bot$\\
$((0,1),a)([0,0],b)$ & $\top$ & $\bot$ & $?(\bot)$\\
$((0,1),a)([2,2],b)$ & $\bot$ & $\bot$ & $\bot$\\
$([0,0],a)([0,0],a)$ & $\bot$ & $\bot$ & $\bot$\\
$([0,0],a)([0,0],b)$ & $\top$ & $\bot$ & $\top$\\
$([0,0],a)([2,2],b)$ & $\top$ & $\bot$ & $?(\bot)$\\
$([0,0],a)([0,0],b)([0,0],b)$ & $\top$ & $\bot$ & $?(\top)$\\
$([0,0],a)([0,0],b)([2,2],b)$ & $\top$ & $\bot$ & $?(\bot)$\\
\bottomrule
\end{tabular}%
\end{tabular}%
} &
\raisebox{-\height}{%
\begin{tabular}[t]{@{}l@{}}
\begin{tabular}[t]{@{}l|cc@{}}
\toprule
$\mathbf{T}_2^8$ & $\epsilon$ & $([0,0],b)$\\ \midrule
$\epsilon$ & $\top$ & $\top$\\
$([0,0],b)$ & $\top$ & $\bot$\\
$([0,0],c)$ & $\bot$ & $\bot$\\
\midrule
$((2,3),b)$ & $\bot$ & $\bot$\\
$([2,2],b)$ & $\top$ & $?(\bot)$\\
$([0,0],b)([0,0],b)$ & $\bot$ & $\bot$\\
$([0,0],b)([0,0],c)$ & $\top$ & $\top$\\
$([0,0],b)([0,0],c)([0,0],b)$ & $\top$ & $?(\bot)$\\
$([0,0],b)([0,0],c)([2,2],b)$ & $\top$ & $?(\bot)$\\
\bottomrule
\end{tabular}\\[45mm]
\begin{tabular}[t]{@{}l@{}}
\begingroup
\setbox0=\hbox{$ u_{\mathit{ce}}^8=(0,a)(2,b)(0,c)(2,b)\in\mathcal{L}(\mathcal{N}) $}%
\setbox1=\hbox{$ \rabs{\project{u_{\mathit{ce}}^8}{1}}=([0,0],a)([2,2],b)([2,2],b) $}%
\setbox2=\hbox{$ \rabs{\project{u_{\mathit{ce}}^8}{2}}=([2,2],b)([0,0],c)([2,2],b) $}%
\dimen0=\wd0
\ifdim\wd1>\dimen0 \dimen0=\wd1\fi
\ifdim\wd2>\dimen0 \dimen0=\wd2\fi
\begin{tabular}[t]{@{}l@{}}
$ u_{\mathit{ce}}^8=(0,a)(2,b)(0,c)(2,b)\in\mathcal{L}(\mathcal{N}) $\\[-1ex]
$\xrightarrow{\makebox[\dimen0][c]{}}$\\[-0.5ex]
$\rabs{\project{u_{\mathit{ce}}^8}{1}}=([0,0],a)([2,2],b)([2,2],b)$\\[0.5ex]
$\rabs{\project{u_{\mathit{ce}}^8}{2}}=([2,2],b)([0,0],c)([2,2],b)$
\end{tabular}
\endgroup
\end{tabular}
\end{tabular}%
}
\end{tabular}%
}
    \caption{TOTUs and counterexample in round 8}
\label{fig:illustrative-totu-round-8}
\end{figure}

%% file: Tables/IllustrativeTOTURound8.tex
\begin{figure}[H]
\centering

\setlength{\tabcolsep}{2.5pt}
\resizebox{0.98\textwidth}{!}{%
\begin{tabular}{@{}l@{\hspace{4mm}}l@{}}
\raisebox{-\height}{%
\begin{tabular}[t]{@{}l@{}}
\begin{tabular}[t]{@{}l|ccc@{}}
\toprule
$\mathbf{T}_1^9$ & $\epsilon$ & $([0,0],a)$ & $([2,2],b)$\\ \midrule
$\epsilon$ & $\top$ & $\top$ & $\bot$\\
$((0,1),a)$ & $\top$ & $\bot$ & $\bot$\\
$((4,5),a)$ & $\bot$ & $\bot$ & $\bot$\\
$([0,0],a)$ & $\top$ & $\bot$ & $\top$\\
\midrule
$([0,0],b)$ & $\bot$ & $\bot$ & $\bot$\\
$((0,1),a)([0,0],a)$ & $\bot$ & $\bot$ & $\bot$\\
$((0,1),a)([0,0],b)$ & $\top$ & $\bot$ & $?(\bot)$\\
$((0,1),a)([2,2],b)$ & $\bot$ & $\bot$ & $\bot$\\
$([0,0],a)([0,0],a)$ & $\bot$ & $\bot$ & $\bot$\\
$([0,0],a)([0,0],b)$ & $\top$ & $\bot$ & $\top$\\
$([0,0],a)([2,2],b)$ & $\top$ & $\bot$ & $\top$\\
$([0,0],a)([0,0],b)([0,0],b)$ & $\top$ & $\bot$ & $?(\top)$\\
$([0,0],a)([0,0],b)([2,2],b)$ & $\top$ & $\bot$ & $?(\top)$\\
$([0,0],a)([2,2],b)([2,2],b)$ & $\top$ & $\bot$ & $?(\top)$\\
\bottomrule
\end{tabular}%
\end{tabular}%
} &
\raisebox{-\height}{%
\begin{tabular}[t]{@{}l@{}}
\begin{tabular}[t]{@{}l|ccc@{}}
\toprule
$\mathbf{T}_2^9$ & $\epsilon$ & $([0,0],b)$ & $([0,0],c)$\\ \midrule
$\epsilon$ & $\top$ & $\top$ & $\bot$\\
$([0,0],b)$ & $\top$ & $\bot$ & $\top$\\
$([0,0],c)$ & $\bot$ & $\bot$ & $\bot$\\
\midrule
$((2,3),b)$ & $\bot$ & $\bot$ & $\bot$\\
$([2,2],b)$ & $\top$ & $?(\bot)$ & $\top$\\
$([0,0],b)([0,0],b)$ & $\bot$ & $\bot$ & $\bot$\\
$([0,0],b)([0,0],c)$ & $\top$ & $\top$ & $\bot$\\
$([2,2],b)([0,0],c)$ & $\top$ & $?(\top)$ & $\bot$\\
$([0,0],b)([0,0],c)([0,0],b)$ & $\top$ & $?(\bot)$ & $\top$\\
$([0,0],b)([0,0],c)([2,2],b)$ & $\top$ & $?(\bot)$ & $\top$\\
$([2,2],b)([0,0],c)([2,2],b)$ & $\top$ & $?(\bot)$ & $\top$\\
\bottomrule
\end{tabular}\\[55mm]
\begin{tabular}[t]{@{}l@{}}
\begingroup
\setbox0=\hbox{$ u_{\mathit{ce}}^9=(1/3,a)(1/12,b)(1/12,c)(23/12,b)\in\mathcal{L}(\mathcal{N}) $}%
\setbox1=\hbox{$ \rabs{\project{u_{\mathit{ce}}^9}{1}}=((0,1),a)((0,1),b)([2,2],b) $}%
\setbox2=\hbox{$ \rabs{\project{u_{\mathit{ce}}^9}{2}}=((0,1),b)((0,1),c)((1,2),b) $}%
\dimen0=\wd0
\ifdim\wd1>\dimen0 \dimen0=\wd1\fi
\ifdim\wd2>\dimen0 \dimen0=\wd2\fi
\begin{tabular}[t]{@{}l@{}}
$ u_{\mathit{ce}}^9=(1/3,a)(1/12,b)(1/12,c)(23/12,b)\in\mathcal{L}(\mathcal{N}) $\\[-1ex]
$\xrightarrow{\makebox[\dimen0][c]{}}$\\[-0.5ex]
$\rabs{\project{u_{\mathit{ce}}^9}{1}}=((0,1),a)((0,1),b)([2,2],b)$\\[0.5ex]
$\rabs{\project{u_{\mathit{ce}}^9}{2}}=((0,1),b)((0,1),c)((1,2),b)$
\end{tabular}
\endgroup
\end{tabular}
\end{tabular}%
}
\end{tabular}%
}
    \caption{TOTUs and counterexample in round 9}
\label{fig:illustrative-totu-round-9}
\end{figure}

%% file: Tables/IllustrativeTOTURound9.tex
\begin{figure}[H]
\centering

\setlength{\tabcolsep}{2.5pt}
\resizebox{0.98\textwidth}{!}{%
\begin{tabular}{@{}l@{\hspace{4mm}}l@{}}
\raisebox{-\height}{%
\begin{tabular}[t]{@{}l@{}}
\begin{tabular}[t]{@{}l|ccc@{}}
\toprule
$\mathbf{T}_1^{10}$ & $\epsilon$ & $([0,0],a)$ & $([2,2],b)$\\ \midrule
$\epsilon$ & $\top$ & $\top$ & $\bot$\\
$((0,1),a)$ & $\top$ & $\bot$ & $\bot$\\
$((4,5),a)$ & $\bot$ & $\bot$ & $\bot$\\
$([0,0],a)$ & $\top$ & $\bot$ & $\top$\\
\midrule
$([0,0],b)$ & $\bot$ & $\bot$ & $\bot$\\
$((0,1),a)([0,0],a)$ & $\bot$ & $\bot$ & $\bot$\\
$((0,1),a)((0,1),b)$ & $\top$ & $\bot$ & $\top$\\
$((0,1),a)([0,0],b)$ & $\top$ & $\bot$ & $\top$\\
$((0,1),a)([2,2],b)$ & $\bot$ & $\bot$ & $\bot$\\
$([0,0],a)([0,0],a)$ & $\bot$ & $\bot$ & $\bot$\\
$([0,0],a)([0,0],b)$ & $\top$ & $\bot$ & $\top$\\
$([0,0],a)([2,2],b)$ & $\top$ & $\bot$ & $\top$\\
$((0,1),a)((0,1),b)([2,2],b)$ & $\top$ & $\bot$ & $?(\top)$\\
$([0,0],a)([0,0],b)([0,0],b)$ & $\top$ & $\bot$ & $?(\top)$\\
$([0,0],a)([0,0],b)([2,2],b)$ & $\top$ & $\bot$ & $?(\top)$\\
$([0,0],a)([2,2],b)([2,2],b)$ & $\top$ & $\bot$ & $?(\top)$\\
\bottomrule
\end{tabular}%
\end{tabular}%
} &
\raisebox{-\height}{%
\begin{tabular}[t]{@{}l@{}}
\begin{tabular}[t]{@{}l|ccc@{}}
\toprule
$\mathbf{T}_2^{10}$ & $\epsilon$ & $([0,0],b)$ & $([0,0],c)$\\ \midrule
$\epsilon$ & $\top$ & $\top$ & $\bot$\\
$([0,0],b)$ & $\top$ & $\bot$ & $\top$\\
$([0,0],c)$ & $\bot$ & $\bot$ & $\bot$\\
\midrule
$((0,1),b)$ & $\top$ & $?(\bot)$ & $\top$\\
$((2,3),b)$ & $\bot$ & $\bot$ & $\bot$\\
$([2,2],b)$ & $\top$ & $?(\bot)$ & $\top$\\
$((0,1),b)((0,1),c)$ & $\top$ & $?(\top)$ & $\bot$\\
$([0,0],b)([0,0],b)$ & $\bot$ & $\bot$ & $\bot$\\
$([0,0],b)([0,0],c)$ & $\top$ & $\top$ & $\bot$\\
$([2,2],b)([0,0],c)$ & $\top$ & $?(\top)$ & $\bot$\\
$((0,1),b)((0,1),c)((1,2),b)$ & $\top$ & $?(\bot)$ & $\top$\\
$([0,0],b)([0,0],c)([0,0],b)$ & $\top$ & $?(\bot)$ & $\top$\\
$([0,0],b)([0,0],c)([2,2],b)$ & $\top$ & $?(\bot)$ & $\top$\\
$([2,2],b)([0,0],c)([2,2],b)$ & $\top$ & $?(\bot)$ & $\top$\\
\bottomrule
\end{tabular}\\[67mm]
\begin{tabular}[t]{@{}l@{}}
$\mathsf{EQ}(\mathcal{H}_1^{10},\mathcal{H}_2^{10})=\top$
\end{tabular}
\end{tabular}%
}
\end{tabular}%
}
    \caption{TOTUs at the final equivalence query}
\label{fig:illustrative-totu-round-10}
\end{figure}